\documentclass[a4paper,UKenglish,cleveref,autoref,thm-restate]{lipics-v2021}

\usepackage{enumitem}
\usepackage[ruled,vlined]{algorithm2e} 
\usepackage{mdframed} 

\crefname{algo}{Algorithm}{algorithms}
\allowdisplaybreaks

\makeatletter
\g@addto@macro\bfseries{\boldmath}
\makeatother

\graphicspath{{fig/}{./fig/}}

\AtBeginDocument{%
  \maketitle
}

\newcommand{\R}{\mathbb{R}}

\newcommand{\rvd}{\mbox{RVD}}            
\newcommand{\pvd}{\mbox{PRVD}}           
\newcommand{\dd}{\mbox{DD}}              
\newcommand{\cR}{\mathcal{S}\xspace}     
\newcommand{\Pol}{\mathcal{P}\xspace}    
\newcommand{\rreg}{vreg}                 
\newcommand{\dreg}{dreg}                 
\newcommand{\rdis}{d_{\angle}}           
\newcommand{\diff}{\mbox{diff}_{\angle}} 
\newcommand{\rbis}{b_{\angle}}           
\newcommand{\dir}{\widehat{d}}           
\newcommand{\apex}{p}                    
\newcommand{\C}{\mathcal{C}}

\newcommand{\Dom}{\mathcal{D}}

\newcommand{\Ham}{\mathcal{H}}     
\newcommand{\dr}{dr}               
  
\newcommand{\etal}{{~et~al.}}
\newcommand{\cb}{C_{b}}            
\newcommand{\minangle}{\alpha^*}   
\newcommand{\minpoint}{x^*}        
\newcommand{\ov}[1]{\overline{#1}} 

\nolinenumbers

\title{The Voronoi Diagram of Rotating Rays with applications to Floodlight Illumination}
\titlerunning{Rotating Rays Voronoi Diagram and Floodlight Illumination}

\author{Carlos Alegr\'{i}a}
{Dipartimento di Ingegneria, Universit\`{a} Roma Tre, Rome, Italy}
{carlos.alegria@uniroma3.it}{https://orcid.org/0000-0002-1825-0097}
{Supported by projects MIUR Proj.\ ``AHeAD'' n\textsuperscript{o} 20174LF3T8 and European Union's Horizon 2020 No. 734922.}

\author{Ioannis Mantas}
{Faculty of Informatics, Università della Svizzera italiana, Lugano, Switzerland}
{ioanni.mantas@gmail.com}{https://orcid.org/0000-0002-1825-0097}
{Initial work supported in part by SNF project 200021E$\_$154387.}

\author{Evanthia Papadopoulou}
{Faculty of Informatics, Università della Svizzera italiana, Lugano, Switzerland}
{evanthia.papadopoulou@usi.ch}{https://orcid.org/0000-0002-1825-0097}
{Supported in part by SNF projects 200021E$\_$154387 and 200021E$\_$201356.}

\author{Marko Savi\'{c}}
{Department of Mathematics and Informatics, Faculty of Sciences, University of Novi Sad, Serbia}
{marko.savic@dmi.uns.ac.rs}{https://orcid.org/0000-0002-1825-0097}
{Supported by the Ministry of Education, Science and Technological Development of the Republic of Serbia (Grant No.\ 451-03-68/2022-14/200125), and the Provincial Secretariat for Higher Education and Scientific Research, Province of Vojvodina.}

\author{Carlos Seara}
{Departament de Matem\`{a}tiques, Universitat Polit\`{e}cnica de Catalunya, Barcelona, Spain}
{carlos.seara@upc.edu}{https://orcid.org/0000-0002-1825-0097}
{Supported by projects
PID2019-104129GB-I00/ MCIN/ AEI/ 10.13039/501100011033 and European Union's Horizon 2020 No. 734922.}

\author{Martin Suderland}
{Faculty of Informatics, Università della Svizzera italiana, Lugano, Switzerland}
{martin.suderland@usi.ch}{https://orcid.org/0000-0002-1825-0097}
{Supported in part by SNF project 200021E$\_$201356.}

\authorrunning{C. Alegr\'{i}a, I. Mantas, E. Papadopoulou, M. Savi\'{c}, C. Seara, and M. Suderland}

\Copyright{Carlos Alegr\'{i}a, Ioannis Mantas, Evanthia Papadopoulou, Marko Savi\'{c}, Carlos Seara, and Martin Suderland}


\category{} 

\relatedversion{Short version in the 29th Annual European Symposium on Algorithms (ESA 2021)\nocite{alegria2021}} 

\supplement{}

\acknowledgements
{
  Initial discussions took place at the \href{https://dccg.upc.edu/irp2018/}{Intensive Research Program in Discrete, Combinatorial and Computational Geometry} in Barcelona, Spain, in 2018.
  We are grateful to the organizers for providing the platform to meet and collaborate.
  We also thank Hendrik Schrezenmaier for early discussions on the problem, mainly related to the results of \cref{subsec:rvdProperties}.
}

\hideLIPIcs  

\EventEditors{John Q. Open and Joan R. Access}
\EventNoEds{2}
\EventLongTitle{42nd Conference on Very Important Topics (CVIT 2016)}
\EventShortTitle{CVIT 2016}
\EventAcronym{CVIT}
\EventYear{2016}
\EventDate{December 24--27, 2016}
\EventLocation{Little Whinging, United Kingdom}
\EventLogo{}
\SeriesVolume{42}
\ArticleNo{23}

\begin{document}

\begin{abstract}
  We study the \emph{Voronoi Diagram of Rotating Rays}, a Voronoi structure where the input sites are rays and the distance function between a point and a site/ray, is the counterclockwise angular distance.
  This novel Voronoi diagram is motivated by illumination or coverage problems, where a domain must be covered by floodlights/wedges of uniform angle, and the goal is to find the minimum angle necessary to cover the domain.
  We study the diagram in the plane, and we present structural properties, combinatorial complexity bounds, and a construction algorithm.
  If the rays are induced by a convex polygon, we show how to construct the Voronoi diagram within this polygon in linear time.
  Using this information, we can find in optimal linear time the \emph{Brocard angle}, the minimum angle required to illuminate a convex polygon with floodlights of uniform angle.
\end{abstract}

\section{Introduction}

In this paper, we propose and study the \emph{rotating rays Voronoi diagram}, which is defined by a set of $n$ rays in the plane under the following oriented angular distance function.
Given a point $x$ and a ray $r$ in the plane, the \emph{angular distance} from $r$ to $x$ is the smallest angle $\alpha$ such that, after a counterclockwise rotation of $r$ around its apex by $\alpha$, the ray $r$ \emph{illuminates} (or touches) $x$; see \cref{fig:distance}.
An example diagram is illustrated in \cref{fig:rvdIntro}.
In this paper we define the diagram, present combinatorial properties and bounds, design construction algorithms, and use the diagram to solve related floodlight illumination problems.

\begin{figure}[t]
  \centering
  \begin{minipage}[t]{0.43\textwidth}
    \centering
    \includegraphics[width=\textwidth,page=2]{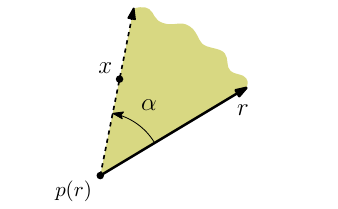}
    \caption
    {
      An $\alpha$-floodlight aligned with a ray $r$ with apex $p(r)$.
      The angle $\alpha$ is the angular distance from $r$ to the point $x$.
    }
    \label{fig:distance}
  \end{minipage}
  \hfill
  \begin{minipage}[t]{0.55\textwidth}
    \centering
    \includegraphics[width=\textwidth,page=2]{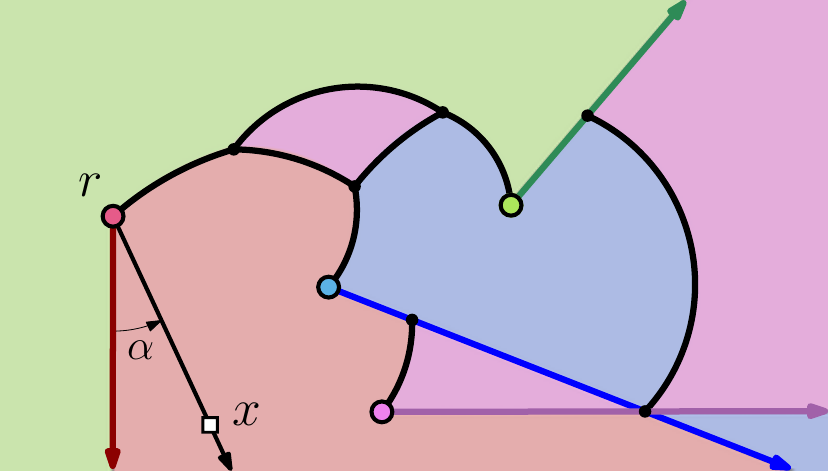}
    \caption
    {
      The rotating rays Voronoi diagram of $4$ rays in $\R^2$.
      The points in each region are first illuminated by the ray of the respective color.
      The angle $\alpha$ is the distance of the point $x \in \R^2$ to its nearest site (ray $r$).
    }
    \label{fig:rvdIntro}
  \end{minipage}
  \label{fig:rvd}
\end{figure}

\subparagraph{Motivation.}

Floodlight illumination problems are well known \emph{art gallery} type of problems, where a given domain has to be covered by \emph{floodlights}, which are light sources illuminating the interior of a cone from its apex.
A floodlight of aperture $\alpha$ is called an \emph{$\alpha$-floodlight}.
An $\alpha$-floodlight is said to be \emph{aligned with} a ray $r$, if the right side of the floodlight coincides with $r$; see for example \cref{fig:distance}.

The angular distance that we consider in this work is motivated by the \emph{Brocard illumination problem}:
Given a domain $\Dom$, a set $\cR$ of $n$ rays, and a set of $n$ $\alpha$-floodlights each aligned with a ray of $\cR$,
what is the minimum angle $\minangle$ required to illuminate $\Dom$ with the set of $\minangle$-floodlights?
The angle $\minangle$ is called the \emph{Brocard angle} of $\Dom$.

In this paper we show that the Brocard illumination problem can be reduced to the construction of the rotating rays Voronoi diagram of $\cR$ restricted to $\Dom$.
The reduction is based on the fact that the Brocard angle is realized at a vertex of the rotating rays Voronoi diagram. 
Typical domains to illuminate by floodlights include 
the plane, bounded polygonal regions, and unbounded regions such as wedges and curves.
Since the construction of the respective Voronoi diagram restricted to each domain yields the Brocard angle, there is an interest in studying the Voronoi diagram in different settings.

\subparagraph{Background and related work.}

The Brocard illumination problem combines floodlight illumination with a generalization of a classic geometric problem first solved by Henri Brocard (1845-1922).
In such problem the input domain is bounded by a triangle, there is a ray aligned with each side of the triangle, and all the rays are oriented either in clockwise or in counterclockwise direction along the boundary of the triangle.
Throughout the years, researchers generalized Brocard's problem to input domains bounded first by convex quadrilaterals, and later on to arbitrary convex polygons; 
see \cref{fig:the_problem1}.
A seminal problem derives from a particular class of convex polygons known as \emph{Brocard polygons}~\cite{bernhart1959}: A polygon $\Pol$ is called a Brocard polygon, if there exists an interior point of $\Pol$ with equal angular distance to all the rays aligned with the edges of $\Pol$.
The angular distance is precisely the Brocard angle of $\Pol$, and the point is known as the \emph{Brocard point} of $\Pol$.
Intuitively speaking, if we continuously increase the value of an angle $\alpha$ starting at $\alpha = 0$, the Brocard point of $\Pol$ is the first one simultaneously illuminated by all the $\alpha$-floodlights.

The characterization of Brocard polygons has a long history, yet, only harmonic polygons (which includes triangles and regular polygons) are known to be Brocard polygons~\cite{casey1888}.
The classic literature on Brocard polygons study the Brocard problem only from a geometric point of view, yet not from a computational perspective.
Nevertheless, from well-known geometric results it is not hard to conclude the following decision result:
{Given a convex polygon $\Pol$ with $n$ vertices, we can decide whether $\Pol$ is a Brocard polygon in $O(n)$ time and, in the affirmative, we can compute the Brocard angle of $\Pol$ in $O(1)$ time}.

\begin{figure}[t]
  \centering
  \begin{subfigure}[t]{0.48\textwidth}
    \centering
    \includegraphics[width=\textwidth,page=1]{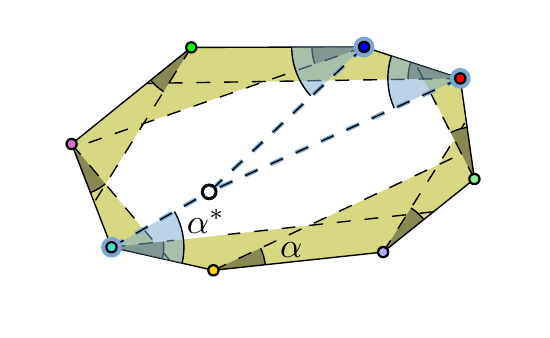}
    \caption
    {
      Illumination of $\Pol$ with an $\alpha$-floodlight aligned with each edge.
    }
    \label{fig:the_problem1}
  \end{subfigure}
  \hfill
  \begin{subfigure}[t]{0.48\textwidth}
    \centering
    \includegraphics[width=\textwidth,page=2]{brocard_illustrate.pdf}
    \caption
    {
      The rotating rays Voronoi diagram of the edge-aligned rays confined into $\Pol$.
    }
    \label{fig:the_problem2}
  \end{subfigure}
  \caption
  {
    A convex polygon $\Pol$.
    Highlighted in blue, the three rays that realize the Brocard angle $\minangle$.    The interior point of $\Pol$ at which $\minangle$ is realized is a rotating rays Voronoi diagram vertex.
  }
  \label{fig:the_problem}
\end{figure}

A natural direction is to consider not only the problem of deciding whether a polygon is Brocard, but the more general problem of computing the Brocard angle of any given polygon.
The problem of computing the Brocard angle of a simple polygon has been recently studied by Alegr\'{i}a\etal~\cite{alegria2017}.
The authors gave an $O(n^3\log^2n)$-time algorithm, and complemented this result with an $O(n\log n)$-time algorithm for convex polygons\footnote{The $O(n)$ time analysis of the algorithm for convex polygons stated in \cite{alegria2017} is not correct.}.
To the best of our knowledge, there are no other studies of the Brocard problem from a computational perspective.

Since their introduction, floodlight illumination problems have been studied in different settings; refer to the book chapters by Urrutia~\cite{urrutia2000} and O'Rourke~\cite{orourke2017} for a compilation of some known results.
Indicatively, the domain may be the entire plane~\cite{bose1997,czyzowicz2015,steiger1998}, an unbounded planar region~\cite{cary2010,steiger1998}, a curve~\cite{contreras1998stage,dietel2008,ito1998,toth2003segments}, or a polygonal domain \cite{estivill1995,ito1998,orourke1987,toth2002}.
The case when floodlights are required to be of uniform angle, as in the Brocard illumination setting, has been explored by several authors, see for example~\cite{contreras1998,estivill1995,ismailescu2008,nilsson_et_al_2021,orourke1995,toth2002,toth2003polygons}.
From a practical point of view, rotating $\alpha$-floodlights can also be used to model devices with limited sensing range (\emph{field of view}), like surveillance cameras or directional antennae; see for example~\cite{berman2007,kranakis2011,neishaboori2014,tao_2015}.
In this context, the Brocard angle is interpreted as the minimum range needed for a set of devices to cover a domain.

Voronoi diagrams are well-studied objects in Computational Geometry with numerous variations and applications; refer to the books of Aurenhammer\etal~\cite{aurenhammerBook} and Okabe\etal~\cite{okabeBook} for a comprehensive list of results.
Still, the rotating rays Voronoi diagram seems to be novel with respect to both the input sites and the distance function.
A slightly related diagram was defined by De Berg\etal~\cite{deBerg2017} to study \emph{dominance regions} of players in the analysis of soccer matches~\cite{taki1996};
such a diagram was also considered recently by Haverkort and Klein~\cite{haverkort2021}.

\subparagraph{Our contribution.}

We introduce the \emph{rotating rays Voronoi diagram} and prove a series of results, paving the way for future work on similar problems.
We define the rotating rays Voronoi diagram restricted to different domains, and show how the Brocard illumination problem in each domain can be reduced to the construction of the corresponding rotating rays Voronoi diagram.
More specifically:
\begin{itemize}[nosep]
\item
  We first consider the diagram of a set of $n$ rays in the plane, and identify structural properties which we complement with complexity results: an $\Omega(n^2)$ worst case lower bound and an $O(n^{2+\epsilon})$ upper bound.
  We also obtain an $O(n^{2+\epsilon})$-time construction algorithm,
  which we use to find the Brocard angle of the plane induced by the set of rays.

\item
  Motivated by the Brocard illumination problem, we study the diagram in a convex polygonal region bounded by the input set of $n$ rays.
  We present a construction algorithm that runs in optimal $\Theta(n)$ time,
  and use this result to find the Brocard angle of a convex polygon in optimal $\Theta(n)$ time.
  This result improves upon the previously known $O(n \log n)$-time algorithm.

\item
  Finally, we study the diagram restricted to simple curves and give a generic approach to  construct it.
  The combinatorial and time complexity bounds depend on the properties of the individual curve. 
\end{itemize}

\subparagraph{Paper outline.}
The rest of the paper is organized as follows.
In \cref{sec:preliminaries} we give the necessary preliminaries.
In \cref{sec:plane} the domain of interest is the entire plane.
In \cref{sec:polygon} we consider a convex polygonal domain, and in \cref{sec:curve} the domain of interest is restricted to curves.
\cref{sec:conclusion} concludes the paper and poses some open questions.

\section{Preliminaries}
\label{sec:preliminaries}

Let $\cR$ be a set of $n$ rays in the plane.
Given a ray $r$, we denote its apex by $\apex(r)$, its supporting line by $l(r)$, and its direction in $S^1$ by $\dir(r)$.
Given three points $A,B,C \in \R^2$, let $\angle(A,B,C)$ denote the counterclockwise angle from $\overrightarrow{BA}$ to $\overrightarrow{BC}$, where $\overrightarrow{BA}$ denote the ray with apex $B$ passing through $A$, and $\overrightarrow{BC}$ respectively.
We define the distance function as follows.
 
\begin{definition}\label{def:distance}
  Given a ray $r$ and a point $x \in \mathbb{R}^2$, the \emph{oriented angular distance} from $x$ to $r$, denoted by $\rdis(x,r)$, is the minimum counterclockwise angle $\alpha$ from $r$ to a ray with apex $\apex(r)$ passing through $x$; see \cref{fig:distance}.
  Further, we define $\rdis(\apex(r),r)=0$.
\end{definition}

It is easy to see that the oriented angular distance (or angular distance, for short) is not a metric.
Moreover, observe that $\rdis(x,r)$ takes values in $[0,2\pi)$ and there is a discontinuity at $2\pi$.
Using this distance function, we can define the bisector of two rays.

\begin{definition}\label{def:dominance}
  Given two rays $r$ and $s$, the \emph{dominance region} of $r$ over $s$, denoted by $\dr(r,s)$, is the locus of points with smaller angular distance to $r$ than to $s$, i.e.,
    \begin{align*}
        && \dr(r,s) := \{\, x \in \R^2 \mid \rdis (x,r) < \rdis (x,s) \,\}.
    \end{align*}
  The \emph{angular bisector} of $r$ and $s$, denoted by $\rbis(r,s)$, is the curve delimiting $\dr(r,s)$ and $\dr(s,r)$; see~\cref{fig:rvdBisectors}.
\end{definition}

Note that because of the discontinuity of the distance function, our definition of a bisector is slightly different than the usual, which is the locus of points equidistant to two sites. 

\begin{figure}[t]
	\centering
	\begin{subfigure}[t]{0.325\textwidth}
		\centering
		\includegraphics[width=\textwidth,page=2]{rvdBisectors.pdf}
		\caption{
			Non-intersecting: $I \notin r,s$.
		}
		\label{fig:rvdBisectors1}	
	\end{subfigure}
	\hfill
	\begin{subfigure}[t]{0.325\textwidth}
		\centering
		\includegraphics[width=\textwidth,page=1]{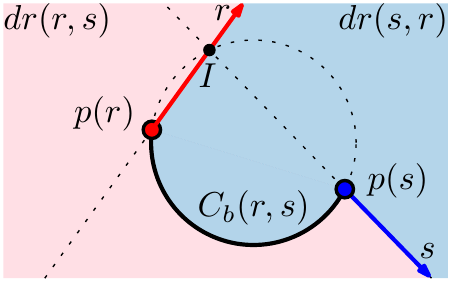}
		\caption{
			Non-intersecting: $I \in r$.
		}		
		\label{fig:rvdBisectors2}		
	\end{subfigure}
	\hfill
	\begin{subfigure}[t]{0.325\textwidth}
		\centering
		\includegraphics[width=\textwidth,page=3]{rvdBisectors.pdf}
		\caption{
			Intersecting.
		}		
		\label{fig:rvdBisectors3}		
	\end{subfigure}\\
	
	\begin{subfigure}[t]{0.325\textwidth}
	\centering
	\includegraphics[width=\textwidth,page=4]{rvdBisectors.pdf}
	\caption{
		"Tangent": $\apex(r) \in s$.
		}
		\label{fig:rvdBisectors4}		
	\end{subfigure}
	\hfill
	\begin{subfigure}[t]{0.325\textwidth}
		\centering
		\includegraphics[width=\textwidth,page=5]{rvdBisectors.pdf}
		\caption{
			"Tangent": $\apex(r) \notin s$.
		}		
		\label{fig:rvdBisectors5}		
	\end{subfigure}
	\hfill
	\begin{subfigure}[t]{0.325\textwidth}
		\centering
		\includegraphics[width=\textwidth,page=6]{rvdBisectors.pdf}
		\caption{
			Sharing their apex.
		}		
		\label{fig:rvdBisectors6}		
	\end{subfigure}\\
	
	\begin{subfigure}[t]{0.325\textwidth}
	\centering
	\includegraphics[width=\textwidth,page=7]{rvdBisectors.pdf}
	\caption{
		Parallel.
	}
		\label{fig:rvdBisectors7}	
	\end{subfigure}
	\hfill
	\begin{subfigure}[t]{0.325\textwidth}
		\centering
		\includegraphics[width=\textwidth,page=8]{rvdBisectors.pdf}
		\caption{
			Anti-parallel: $l(r) \neq l(s)$.
		}		
		\label{fig:rvdBisectors8}		
	\end{subfigure}
	\hfill
	\begin{subfigure}[t]{0.325\textwidth}
		\centering
		\includegraphics[width=\textwidth,page=9]{rvdBisectors.pdf}
		\caption{
			Anti-parallel: $l(r) = l(s)$.
		}	
		\label{fig:rvdBisectors9}			
	\end{subfigure}
  \caption
  {
    The angular bisector of two rays $r$ and $s$.
    The bisector consists of $r$ (red ray), $s$ (blue ray), and an arc of the bisecting circle $\cb(r,s)$ (black curve).
  }
  \label{fig:rvdBisectors}
\end{figure}

Given two rays $r$ and $s$, the bisector $\rbis(r,s)$ consists of the two rays $r$ and $s$, and a circular arc $a$ that connects $\apex(r)$ to $\apex(s)$;
see \cref{fig:rvdBisectors}.
Let $I := l(r) \cap l(s)$.
The arc $a$ belongs to the \emph{bisecting circle} $\cb(r,s)$, which we define as follows:
\begin{itemize}[leftmargin=*,topsep=1mm]
\item
  If $I$, $\apex(r)$, and $\apex(s)$ are pairwise different, then $\cb(r,s)$ is the circle through $I$, $\apex(r)$, and $\apex(s)$.
  The arc $a$ contains $I$ if, and only if, either $I$ lies on both $r$ and $s$, or $I$ lies on none of $r$ and $s$.
  See \cref{fig:rvdBisectors1}, \cref{fig:rvdBisectors2}, and \cref{fig:rvdBisectors3}.

\item
  If $I=\apex(r)$ and $I \neq \apex(s)$, then $\cb(r,s)$ is the circle tangent to $l(r)$ passing through $\apex(r)$ and $\apex(s)$.
  Both $a$ and $r$ lie on the same side of $l(s)$, if and only if, $p(r)$ lies on $s$; see \cref{fig:rvdBisectors4} and \cref{fig:rvdBisectors5}.
  We analogously define $\cb(r,s)$, if $I=\apex(s)$ and $I \neq \apex(r)$.

\item
  If $\apex(r) = \apex(s)$, then both $\cb(r,s)$ and $a$ degenerate to a single point; see \cref{fig:rvdBisectors6}.

\item
  If $l(r)$ and $l(s)$ are parallel, then $\cb(r,s)$ degenerates to the line through $\apex(r)$ and $\apex(s)$.
  If $\dir(r) = \dir(s)$, then $a$ consists of two halflines; see \cref{fig:rvdBisectors7}.
  If instead $\dir(r) = -\dir(s)$, then $a$ degenerates to a line segment; see \cref{fig:rvdBisectors8} and \cref{fig:rvdBisectors9}.
  
\end{itemize}

Unless otherwise stated, we assume for simplicity that no two rays share an apex and no two rays are parallel or antiparallel.
Under these assumptions, the bisectors are of the forms illustrated in \cref{fig:rvdBisectors1,fig:rvdBisectors2,fig:rvdBisectors3,fig:rvdBisectors4,fig:rvdBisectors5}.
In the following lemma we establish the correctness of the description of the bisectors as described above.

\begin{lemma}\label{lem:NonParallelBisector}
  Given two rays $r$ and $s$,
  the bisector $\rbis(r,s)$ consists of the two rays $r,s$ and an arc of the bisecting circle $\cb(r,s)$.
\end{lemma}

\begin{proof}
  Any point slightly to the left of ray $r$ has a distance 
  of almost $0$
  to ray $r$, whereas any point slightly to the right of ray $r$ has a distance 
  of almost $2 \pi$.
  Hence the rays $r$ and $s$ are part of the bisector $\rbis(r,s)$.
  
\begin{figure}[t]
    \centering
    \begin{subfigure}[t]{0.48\textwidth}
        \centering
        \includegraphics[width=0.9\textwidth,page=1]{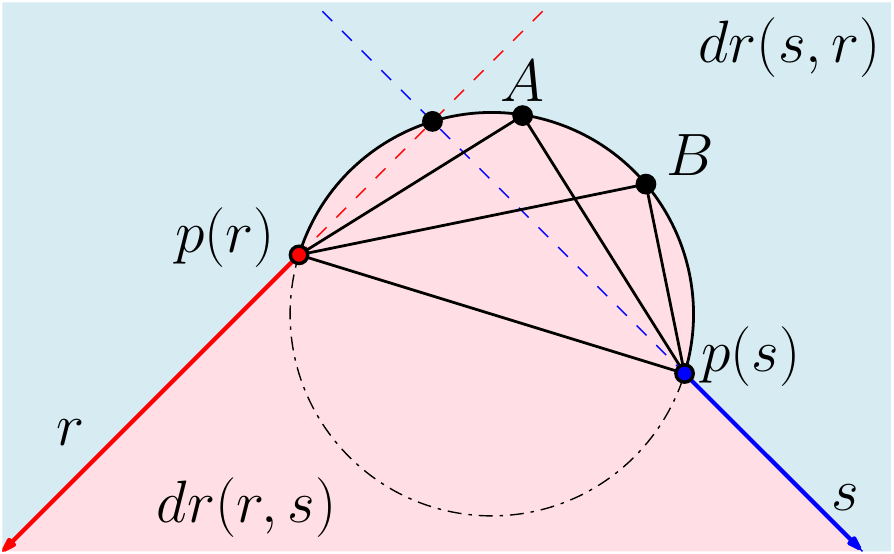}
        \caption{
            Any two points equidistant to both rays $r$ and $s$ lie on a common circle.
        }
        \label{fig:proof_bisector1}
    \end{subfigure}
    \hfill
    \begin{subfigure}[t]{0.48\textwidth}
        \centering
        \includegraphics[width=0.9\textwidth,page=2]{fig/rvdBisectorProof.pdf}
        \caption{
            The intersection point $I$ lies on the common circle with all the points equidistant to $r$ and $s$.
        }
        \label{fig:proof_bisector2}
    \end{subfigure}
    \caption{
        Illustrations for the proof of \cref{lem:NonParallelBisector}.
    }
    \label{fig:proof_bisector}
\end{figure} 
  We first show that any two points $A$ and $B$ 
  equidistant to both rays $r$ and $s$ lie on a common circle; see \cref{fig:proof_bisector1}.
  Since $A,B$ are equidistant to $r$ and $s$, it 
  means that $\angle(B,p(r),A) = \angle(B,p(s),A)$.
  We show that $\angle(p(r),A,p(s)) = \angle(p(r),B,p(s))$, which implies that all $p(r),p(s), A$ and $B$ lie on a circular arc connecting $p(r)$ and $p(s)$ by the inscribed angle theorem:
  \begin{align*}
    &\angle(p(r),A,p(s)) \\
    & = \pi - \angle(p(s),p(r),A) - \angle(A,p(s),p(r))\\
    & = \pi-(\angle(p(s),p(r),B)+\angle(B,p(r),A))- (\angle(B,p(s),p(r))-\angle(B,p(s),A))\\
    & = \pi - \angle(p(s),p(r),B) - \angle(B,p(s),p(r))\\
    & = \angle(p(r),B,p(s)).
  \end{align*}
  
  In the final step we show that $I = l(r) \cap l(s)$ lies on the common circle with all the equidistant points.
  If $I$ lies on both (resp. none) of the rays $r$ and $s$ then $\rdis(I,r) = \rdis(I,s) = 0$ (resp. $\rdis(I,r) = \rdis(I,s) = \pi$).
  In this case $I$ is equidistant to both rays and therefore, clearly on the common circle.
  
  Let us now assume that $I$ lies on exactly one of the rays $r$ and $s$; see \cref{fig:proof_bisector2}.
  Let $A$ be a point equidistant to both rays, i.e., $\angle(I,p(r),A) = \pi + \angle(I,p(s),A)$.
  Then
  \begin{align*}
    \angle(p(s),A,p(r)) 
    & = 2 \pi -\angle(A,p(r),I)-\angle(p(r),I,p(s))-\angle(I,p(s),A)\\
    & = \pi - \angle(p(r),I,p(s)).
  \end{align*}
  Therefore, by the inscribed angle theorem, $A$, $p(r)$, $I$ and $p(s)$ lie on opposite sides of a common circle, concluding the proof.
\end{proof}

\begin{definition}\label{def:angularDifference}
  The \emph{angular difference} between two rays $r$ and $s$, denoted by $\diff(r,s)$, 
   is the angle by which $s$ has to rotate counterclockwise around its apex,  
  so that $r$ and $s$ become parallel.
\end{definition}

The angular difference of two rays is illustrated in \cref{fig:rvdAngularDifference}.
Note that for any two non-parallel rays $r$ and $s$ we have that $\diff(r,s) + \diff(s,r) = 2\pi = 0$.

\begin{remark}
    \label{remark:distance}
    Given a pair of rays $r$ and $s$, the distance function is monotone along the circular arc of their bisector $\rbis(r,s)$, and strictly monotone if the lines $l(r)$ and $l(s)$ are not parallel.
    If the lines $l(r)$ and $l(s)$ are parallel,
    then the distance is constant along the entire circular part of the bisector $\rbis(r,s)$. 
    
    If instead $\diff(r,s) < \diff(s,r)$, or equivalently $\diff(s,r) > \pi$, then the distance function along the bisector $\rbis(r,s)$ from $p(r)$ to $p(s)$ is monotonically increasing.
    Moreover, walking along the boundary of $\dreg(r,s)$ in counterclockwise order, the distance function on the circular part of bisector $\rbis(r,s)$ is 
    monotonically increasing
    (see the arrow in \cref{fig:rvdAngularDifference}).
\end{remark}

We can now define the nearest Voronoi diagram of a set of rays under the 
angular distance. 

\begin{definition}
  The \emph{Rotating Rays Voronoi Diagram} of a set $\cR$ of rays is the subdivision of $\mathbb{R}^2$ into 
  \emph{Voronoi regions} defined as follows:
  \begin{linenomath*}
  \begin{align*}
    &&     \rreg(r) := \{\, x \in \mathbb{R}^2 \mid \forall s \in \cR \setminus\{r\}: \ \rdis(x,r) < \rdis(x,s) \, \}.
  \end{align*}
  \end{linenomath*}  
  Let $\rvd(\cR) := \left(\mathbb{R}^2 \setminus \bigcup_{r \in \cR} \rreg(r)\right) \cup \cR$ denote the graph structure of the diagram. 
  \label{def:diagram}
\end{definition}

A Voronoi region $\rreg(r)$ can be equivalently defined as the intersection of all the dominance regions of $r$, i.e., $\rreg(r) = \bigcap_{s \in \cR \setminus \{r\}} \dr(r,s)$.
A region may consist of more than one connected components;
each component is a called a \emph{face} of the region.

We distinguish the following features on $\rvd(\cR)$.
Refer to \cref{fig:rvdFeatures} for an illustration.

\begin{itemize}[nosep]
	\item
	A \emph{circular edge} is a subset of the circular part of a bisector, thus,
	any point on a circular edge is equidistant to the two sites that induce it
	(see $\ov{vw}$ in \cref{fig:rvdFeatures}).
	
	\item
	A \emph{ray edge} is a subset of a ray, thus, 
	any point on a ray edge has distance $0$ to the site that induces it
	(see $\ov{xw}$ in \cref{fig:rvdFeatures}).
	
	\item
	A \emph{proper vertex} is incident to three circular edges, thus,
	it is equidistant to the three sites that induce the three circular edges
	(see $u$ in \cref{fig:rvdFeatures}).

	\item
	A \emph{mixed vertex} is incident to one circular edge and two ray edges, which are induced by a single ray.
	It is equidistant to the two sites inducing the circular edge and 
	has distance $0$ to the site inducing the ray edges.
	(see $v$ in \cref{fig:rvdFeatures}).

	\item
	An \emph{intersection vertex} is incident to one circular edge and four ray edges, all of which are induced by two sites.
	It has distance $0$ to both sites
	(see $w$ in \cref{fig:rvdFeatures}).
	
	\item
	An \emph{apex vertex} is incident to one circular edge and one ray edge, where the site inducing the ray edge is one of the two sites inducing the circular edge.
	It has distance $0$ to the site inducing the ray edge 
	and distance greater than $0$ to the other site
	(see $x$ in \cref{fig:rvdFeatures}).
\end{itemize}

\begin{figure}[t]
	\begin{minipage}[t]{0.33\textwidth}
		\centering
		\includegraphics[width=0.95\textwidth,page=12]{rvdBisectors.pdf}
		\caption{
		The angular difference $\diff(r,s) = \beta$.
		The distance is increasing from $\apex(r)$ to $\apex(s)$ along bisector $\rbis(r,s)$.
		}
		\label{fig:rvdAngularDifference}
	\end{minipage}	
	\hfill
	\begin{minipage}[t]{0.64\textwidth}
	\centering
	\includegraphics[width=0.9\textwidth,page=1]{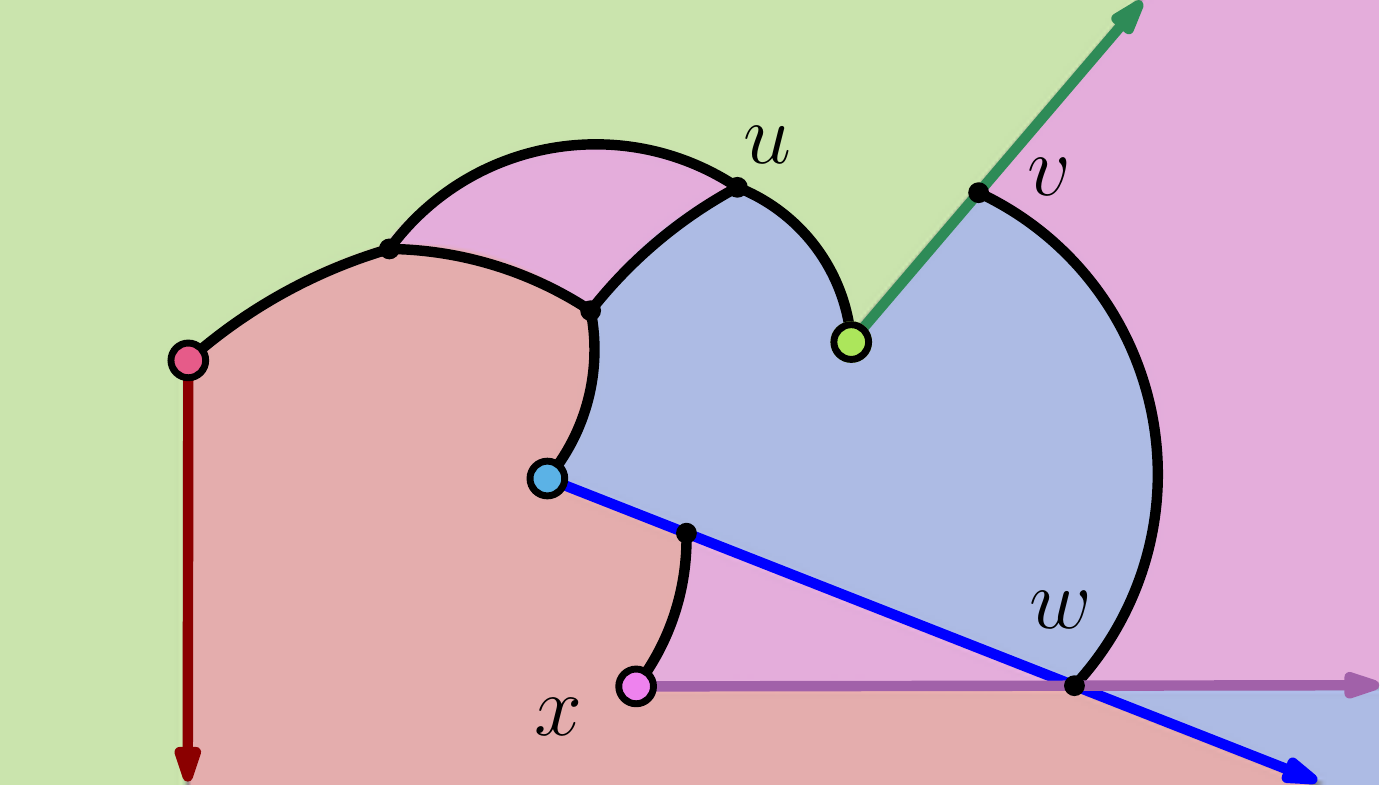}
	\caption{
		Different features on $\rvd(\cR)$:
		arc $\ov{vw}$ is a circular edge, segment $\ov{xw}$ is a ray edge,
		$u$ is a proper vertex, $v$ is a mixed vertex, $w$ is an intersection vertex, and $x$ is an apex vertex.
	}
	\label{fig:rvdFeatures}
	\end{minipage}
\end{figure}

\section{Rotating Rays Voronoi diagram in the Plane}
\label{sec:plane}

In this section we study the diagram $\rvd(\cR)$ in the plane.
We first look at some properties and combinatorial complexity bounds.
Then we consider the problem of illuminating the plane with a set of floodlights aligned with $\cR$.

\subsection{Properties, complexity, and a construction algorithm}
\label{subsec:rvdProperties}

We first study the structure of the Voronoi diagram of $3$ rays; see an example in \cref{fig:5mixedvertices}.

\begin{lemma}
\label{lem:diagramOfThree}
The Voronoi diagram of three rays $\rvd(\{r,s,t\})$ has at most $1$ proper Voronoi vertex, at most $3$ intersection vertices, and at most $6$ mixed vertices. 
Its overall combinatorial complexity is $O(1)$.
\end{lemma}
\begin{proof}
  A proper Voronoi vertex is 
  the intersection point of three circular arcs of three related bisectors.
  Consider two related bisectors $\rbis(r,s)$ and $\rbis(r,t)$, and their bisecting circles $\cb(r,s)$ and $\cb(r,t)$. 
  Two circles intersect at most twice, and the circles $\cb(r,s)$ and $\cb(r,t)$ already have one point of intersection by definition, i.e., apex $\apex(r)$.
  Hence, there there can be at most one more point of intersection $v$ between them, and consequently between the circular arcs of $\rbis(r,s)$ and $\rbis(r,t)$.
  Given an intersection point $v$, the bisector $\rbis(s,t)$ also passes by $v$, inducing a proper vertex of $\rvd(\{r,s,t\})$ at point $v$. So, there exists at most one proper vertex in $\rvd(\{r,s,t\})$.
  
  An intersection vertex is defined at the intersection point of two rays. So, given three rays there are at most $3$ intersection vertices in $\rvd(\{r,s,t\})$.
  A mixed vertex is defined at the intersection of a circular arc of a bisector $\rbis(r,s)$ and of a ray $t \notin \{r,s\}$. 
  Consider the circular arc of the $\rbis(r,s)$; 
  the remaining ray $t$ can intersect the arc at most two times, hence inducing at most two mixed vertices on the circular arc of $\rbis(r,s)$.
  Given three rays, there are three bisectors, so overall there can be no more than $6$ mixed vertices in $\rvd(\{r,s,t\})$.

  There are $O(1)$ vertices in $\rvd(\{r,s,t\})$, so the combinatorial complexity follows. 
\end{proof}

\begin{figure}[t]
        \centering
        \includegraphics[width=0.6\textwidth,page=5]{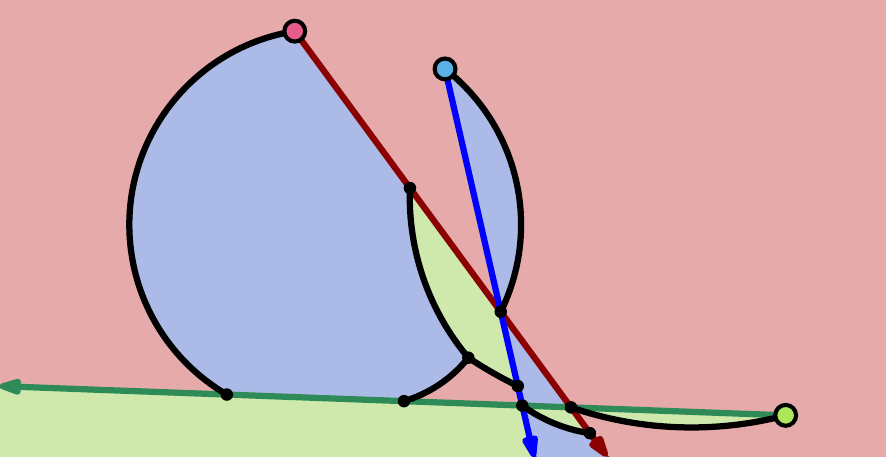}
        \caption{
            A set $\cR$ of $3$ rays, with $\rvd(\cR)$ having $1$ proper, $5$ mixed, and $3$ intersection vertices.
        }
    \label{fig:5mixedvertices}
\end{figure}

Assuming that no two rays of $\cR$ are parallel to each other, 
the following two simple structural properties hold.

\begin{lemma}
    \label{lem:oneUnbounded}
    $\rvd(\cR)$ has exactly $n$ unbounded faces, one for each ray; 
    each ray is incident to its unbounded face.
\end{lemma} 
\begin{proof}
    Let $C$ be a circle of sufficiently large radius so that $C$ encloses all vertices of $\rvd(\cR)$ and the bisecting circles of all bisectors.
    To study the unbounded faces of a Voronoi region $\rreg(r)$, for a ray $r\in \cR$,
	we consider the intersection of $\rreg(r)$ with $C$.
	Refer to \cref{fig:rvdPropertiesPlane1} for an illustration.
 
    Given a ray $s \in \cR \setminus \{r\}$, the intersection of the dominance region $\dr(r, s)$ with $C$ is a circular arc on $C$ lying counterclockwise from $r \cap C$ to $s \cap C$.
    Since the region $\rreg(r)$ is the intersection of all the dominance regions of $r$, it follows that $\rreg(r) \cap C$ is the intersection of $n-1$ circular arcs, all starting from $r$.
    Hence, $\rreg(r) \cap C$ is a non-empty circular arc,
    incident to $r$.
    Thus, $\rreg(r)$ has exactly one unbounded face incident to $r$.
\end{proof}

\begin{figure}
	\centering
	\begin{minipage}[t]{0.48\textwidth}
		\centering
		\includegraphics[width=0.58\textwidth,page=1]{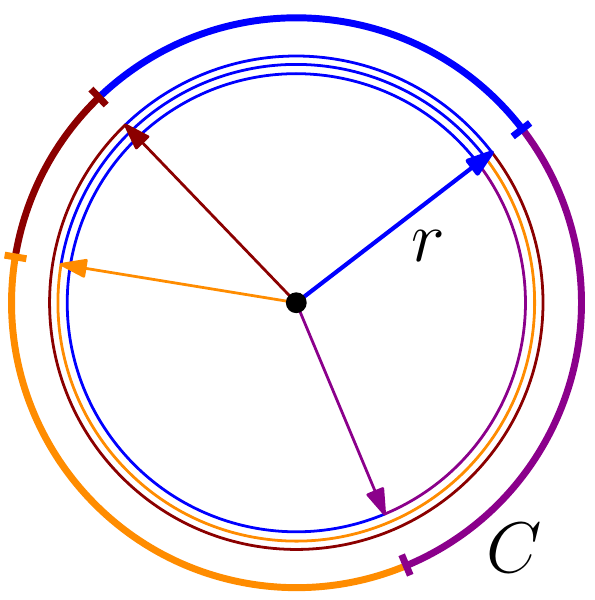}
		\caption{
			Intersection of a diagram of 4 rays with a large circle $C$.
			Dominance regions are circular arcs on $C$.
		}
		\label{fig:rvdPropertiesPlane1}	
	\end{minipage}
	\hfill	
	\begin{minipage}[t]{0.48\textwidth}
		\centering
		\includegraphics[width=.85\textwidth,page=2]{rvdProperties.pdf}
		\caption{
			Two (impossible) cases leading to a disconnected diagram.
			A ``corridor'' ($\rreg(r_1)$) and
			an ``island'' ($\rreg(r_2)$).
		}
		\label{fig:rvdPropertiesPlane2}	
	\end{minipage}
\end{figure}

\begin{lemma}
  \label{lem:connectedDiagram}
  $\rvd(\cR)$ is connected.
\end{lemma}
\begin{proof}
    Assume, to the contrary, that $\rvd(\cR)$ is not connected, as illustrated in \cref{fig:rvdPropertiesPlane2}.
    Then, there exists a ray $r \in \cR$ whose region $\rreg(r)$ either 
    has an unbounded face with two occurrences at infinity, having conceptually a \emph{``corridor''} (as $\rreg(r_1)$ in \cref{fig:rvdPropertiesPlane2}), or $\rreg(r)$ contains  a connected component  of $\rvd(\cR)$, creating conceptually an \emph{``island''} (as $\rreg(r_2)$ in \cref{fig:rvdPropertiesPlane2}).
    
    By \cref{lem:oneUnbounded}, it directly follows that no region of $\rvd(\cR)$ can have a ``corridor''.
    We prove that no region $\rreg(r)$ of $\rvd(\cR)$ contains an ``island'', i.e., a connected component of $\rvd(\cR)$ entirely surrounded by $\rreg(r)$.
    Consider such a disconnected component of $\rvd(\cR)$ surrounded by $\rreg(r)$; this component contains at least one face of a region $\rreg(s)$ for some $s \in \cR$.  
    Then, also in $\rvd(\{r,s\})$, there is an ``island'' inside $\rreg(r)$, implying that the bisector $\rbis(r,s)$ has a bounded connected component.
    This contradicts the fact that each bisector is a simple unbounded curve.
\end{proof}

We now study the combinatorial complexity of $\rvd(\cR)$.
An $\Omega(n^2)$ lower bound is easily 
derived by a set $\cR$ of $n$ pairwise intersecting rays.
In such case,
$\rvd(\cR)$ has $\binom{n}{2}=\Theta(n^2)$ vertices (at the intersection of rays) and thus $\Omega(n^2)$ complexity. 
Interestingly, this bound holds even for 
non-intersecting rays, as we will show in the following theorem.

\begin{theorem}
\label{thm:lowerBound}
  The worst case combinatorial complexity of $\rvd(\cR)$ has an $\Omega(n^2)$ lower bound, even if the rays are pairwise non-intersecting.
\end{theorem}
\begin{proof}
  We give a constructive proof; the resulting diagram is illustrated in \cref{fig:lowerBoundsDiagram}. The  Voronoi regions of the $n/2 -1$ rays with the leftmost apices have $n/2$ bounded faces each.

  We set $n = 2m$ and let the apices $p(r_i) = (i,0)$, $i=1,\dotsc,2m$.
  For $i=m+1,\dotsc,2m$, let the direction of $r_i$ be vertically upwards. 
  For $i=1,\dotsc,m$, let the direction of $r_i$ be $\dir(r_i)=(\sin \alpha_i, \cos \alpha_i)$ with $\alpha_1 \in (3\pi/2,2\pi)$ and $\alpha_i = \alpha_{i-1} + \epsilon_i$ where $\epsilon_i > 0$ for $i=2,\dotsc,m$.
  We choose $\epsilon_i$ one by one, in the increasing order of $i$, so that both $r_i$ and $r_{i-1}$ have a face between any two consecutive upward shooting rays.
  This is always possible since we can choose $\epsilon_i$ small enough so that, at any $x$-coordinate $x < 2m$, the circular part of $\rbis(r_i, r_{i-1})$ is arbitrarily close to the $x$-axis, and thus, is below the circular part of $\rbis(r_{i-2}, r_{i-1})$.
\end{proof}

\begin{figure}[t]
    \centering
    \begin{minipage}[t]{0.45\textwidth}
        \centering
	    \includegraphics[width=\textwidth,page=1]{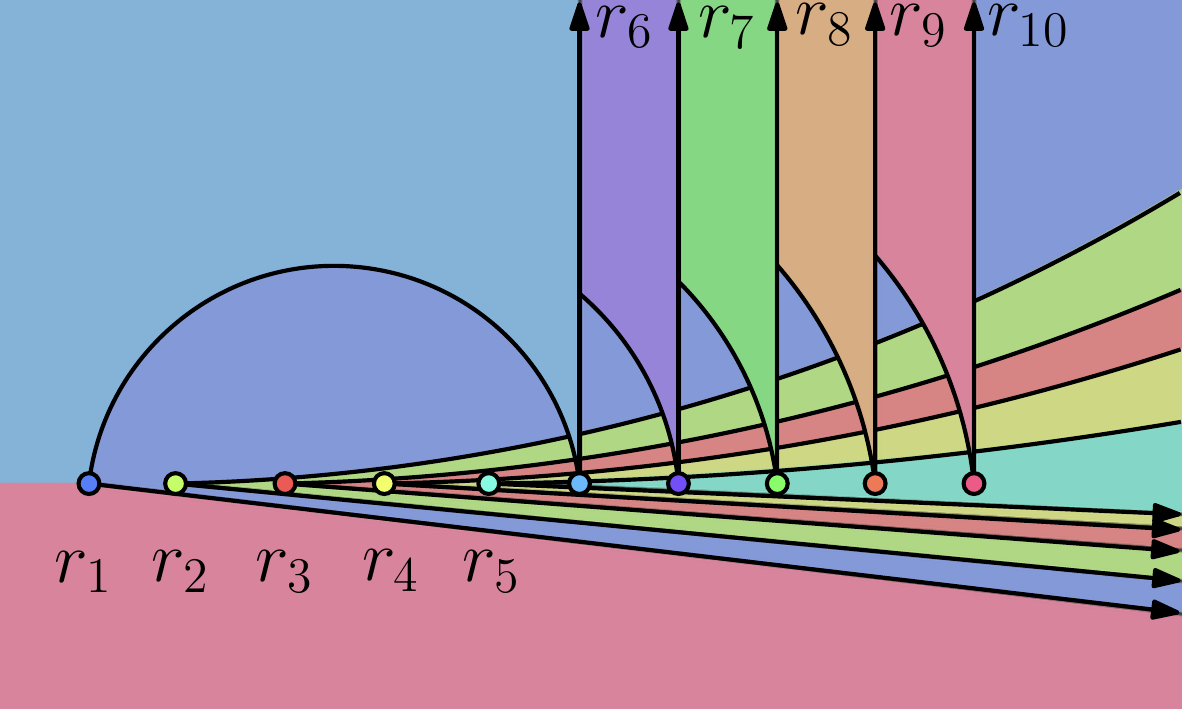}
        \caption{
    		A set $\cR$ of $n=10$ pairwise non-intersecting rays with $\rvd(\cR)$ having $\Theta(n^2)$ complexity.
    		The region $\rreg(r_i)$, $i=1, \dots, 4$, has $\Theta(n)$ faces.            
        }
        \label{fig:lowerBoundsDiagram}
    \end{minipage}
    \hfill
    \begin{minipage}[t]{0.53\textwidth}
        \centering
		\includegraphics[trim=10 0 0 0, clip, width=\textwidth,page=1]{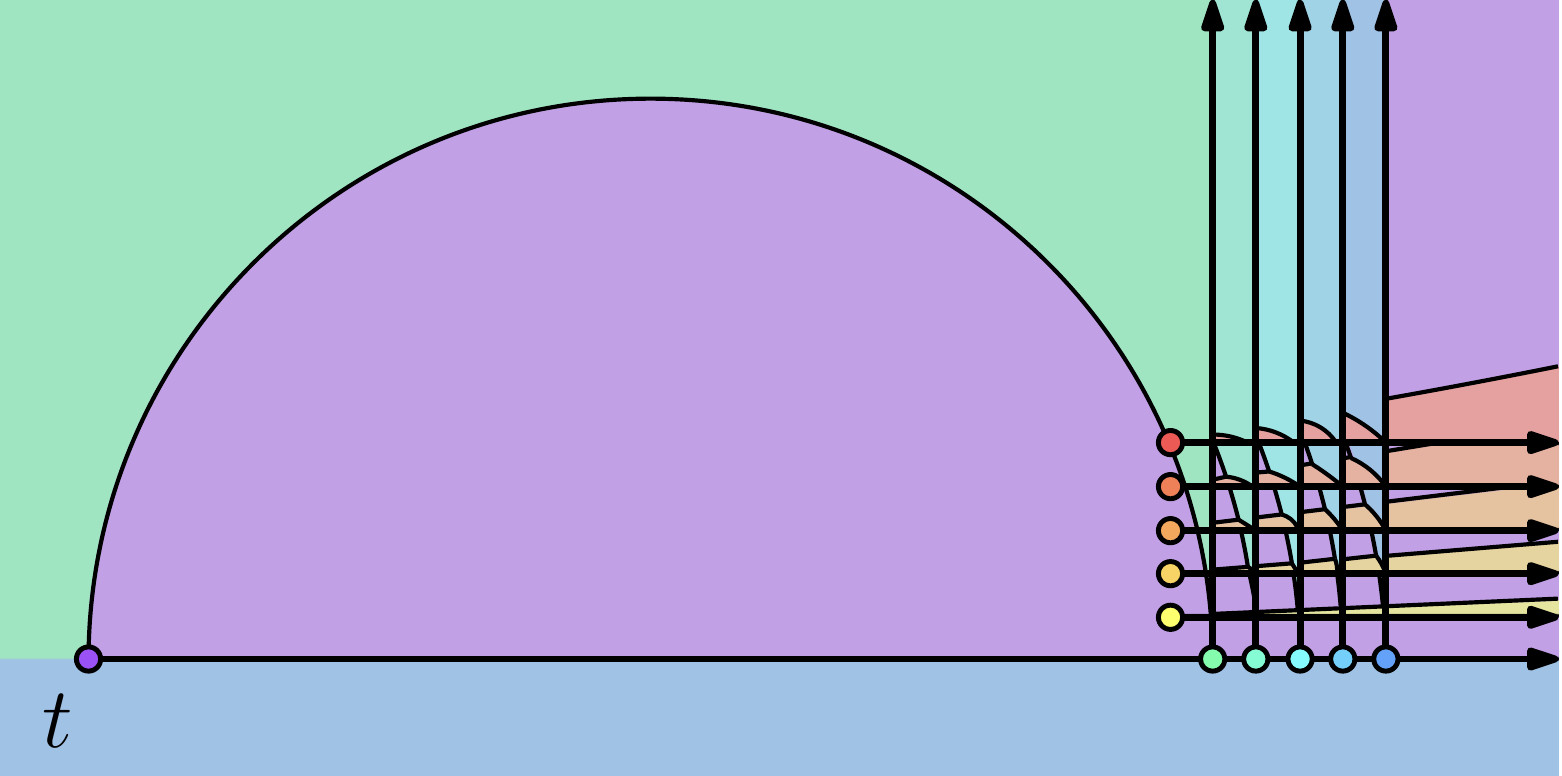}    
        \caption{
		A set $\cR$ of $n=11$ rays with $\rvd(\cR)$.
		The region $\rreg(t)$ has $\Theta(n^2)$ faces, one in each cell of the grid formed by
		the other 10 rays.
        }
        \label{fig:lowerBoundsRegion}
    \end{minipage}
\end{figure}

\begin{theorem}
  \label{lem:lowerBoundFace}
  A Voronoi region of $\rvd(\cR)$ has $\Theta(n^2)$ complexity in the worst case.
\end{theorem}
\begin{proof}
  Consider a ray $r \in \cR$ and its region $\rreg(r)$; all but at most $O(n)$ vertices (possible apex vertices) on the boundary of $\rreg(r)$ are defined by $r$ and a pair of other sites.
  There are $\Theta(n^2)$ pairs in $\cR \setminus \{r\}$, each inducing $O(1)$ vertices on $\rreg(r)$ (by \cref{lem:diagramOfThree}); so $\rreg(r)$ has $O(n^2)$ vertices, and thus, $O(n^2)$ combinatorial complexity.

  We now give a construction of $n=2m+1$ rays, where a single region has $\Theta(n^2)$ complexity; refer to the construction shown in \cref{fig:lowerBoundsRegion}.
  We first create a grid structure:
  for $i = 1,\ldots,m$, let $r_i$ be a ray with $p(r_i)= (i,0)$ \emph{shooting} vertically upward and let $s_i$ be a ray with $p(s_i)= (0,i)$ \emph{shooting} horizontally to the right.
  For all $(i,j) \in \{1, \ldots, m-1\}^2$, let $R(i,j)$ be the square $[i,i+1) \times [j,j+1)$.
  Each~square $R(i,j)$ is made up of two faces of $\rvd(\{r_1,\ldots,r_m,s_1,\ldots,s_m\})$, one~belonging to $\rreg(r_i)$ and one belonging to $\rreg(s_j)$.
  Now let $\alpha(i,j) := \allowbreak \max\{\, \min\{\rdis(x, r_i), \rdis(x, s_j)\} \mid x \in R(i,j)\,\}$, and let $\alpha_{\min} :=  \min\{\, \alpha(i,j) \mid (i,j) \in \{1, \ldots, m-1\}^2\,\}$.
  It is easy to see that $\alpha_{\min} > \arctan 1/m$.

  We now introduce another ray $t$, so that $\max\{\,\rdis(x,t) \mid x \in [1,m]^2\,\} < \alpha_{\min}$.
  One way to achieve this is to set $\apex(t) = (-m^2,0)$ and make $t$ \emph{shooting} horizontally to the right.
  This means that in each $R(i,j)$, for $(i,j) \in \{1, \ldots, m-1\}^2$, there is a point which will be visited by the ray $t$ before it is visited before any of the rays $r_i$ or $s_j$, meaning the region $\rreg(t)$ has a face in each square, which is $\Theta(n^2)$ faces in total.
\end{proof}

The above theorem directly implies an $O(n^3)$ upper bound on the complexity of $\rvd(\cR)$.
Next we 
show how the angular distance function can be adapted so that we can apply 
the general upper bounds of Sharir~\cite{sharir1994}.
As a by-product, we also obtain a construction algorithm for $\rvd(\cR)$.

\begin{theorem}\label{thm:rvdGeneralUpperBound}
  For any $\epsilon > 0$, $\rvd(\cR)$ has $O(n^{2+\epsilon})$ combinatorial complexity.
  Further, $\rvd(\cR)$ can be constructed in $O(n^{2+\epsilon})$ time.
\end{theorem}

\begin{proof}
  Each site $r$ induces a \emph{distance function} $d^r_\angle(x) := \rdis(x,r)$ which maps a point $x = (x_1,x_2) \in \mathbb{R}^2$ to its angular distance from $r$.
  Consider the lower envelope of the graphs of these distance functions in $3$-space.
  The diagram $\rvd(\cR)$ can be seen as the projection of this lower envelope to the plane.
For algebraic distance functions, Sharir~\cite{sharir1994} gives complexity bounds for this lower envelope accompanied with algorithmic results. The angular distance functions though are not algebraic. Our strategy is to apply the result from \cite{sharir1994} to functions $d^r_{\text{alg}}$ that are equivalent to the functions $d^r_{\angle}$ in the sense that the both set of functions would produce the same lower envelope, but each function $d^r_{\text{alg}}$ is made of a constant number of algebraic surface patches in $3$-space. More precisely, we want to find piece-wise algebraic functions $d^r_{\text{alg}}$ that fulfill the following property:
\begin{linenomath*}
\begin{align*}
    &&
    d^r_\angle(x)<d^s_\angle(x) \, \Longleftrightarrow \, d^r_{\text{alg}}(x) < d^s_{\text{alg}}(x)
\end{align*}
\end{linenomath*}
for all $r,s \in \cR$ and $x \in \R^2$.
  
  Without loss of generality, assume that $\apex(r)$ lies on the origin and $r$ is facing to the right in positive $x_1$-direction of the coordinate system.
  Let $x\in\mathbb{R}^2$ and $\alpha := d^r_\angle(x)$.
  Then we want to set $d^r_{\text{alg}}(x) := 1 - \cos(\alpha)$ if $0 \leq \alpha \leq \pi$, and $d^r_{\text{alg}}(x) := 3 + \cos(\alpha)$ if $\pi \leq \alpha < 2\pi$.
  The function $x \mapsto \cos(\alpha)$ is indeed algebraic since it is obtained by first scaling $x$ to unit length and then mapping it to its first coordinate.
  Then we have
  \begin{linenomath*}
  \begin{align*}
    && d^r_{\text{alg}}((x_1,x_2)) = 
    \begin{cases}
     0                             & \text{if } x_1=x_2=0, \\
     1 - \frac{x_1}{\sqrt{x_1^2+x_2^2}}  & \text{if } x_1\neq 0,\, x_2 \geq 0, \\
     3 + \frac{ x_1}{\sqrt{x_1^2+x_2^2}} & \text{otherwise.}
    \end{cases}
  \end{align*}
  \end{linenomath*}
  Since $d^r_{\text{alg}}$ consists of three patches, which are all algebraic and have simple domain boundaries, applying \cite{sharir1994} to these functions yields the claimed combinatorial and algorithmic results.
\end{proof}

\subsection{Brocard illumination of the plane}
\label{subsec:rvd_illumination}

We now look into the Brocard illumination problem in $\R^2$.
Recall that given a set of rays $\cR$, and an $\alpha$-floodlight aligned with each ray, 
the problem asks for the Brocard angle which is the minimum angle needed to illuminate a target domain. 
The Brocard angle of $\R^2$ is 
\begin{align*}
    &&
    \minangle = \max_{x \in \R^2}\min_{r \in \cR}\rdis(x,r).
\end{align*}

    \begin{figure}[b]
	\centering
	\begin{subfigure}[t]{0.48\textwidth}
		\centering
		\includegraphics[width=.92\textwidth,page=21]{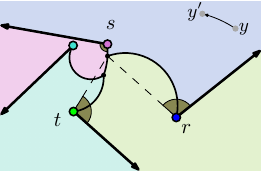}
		\caption{
			$\minangle$ is realized on a vertex of $\rvd(\cR)$, by rays $r,s,t$.
			Point $y'$ is further than $y$ to its nearest ray. 
   }
		\label{fig:rvdBrocardPlane1}
	\end{subfigure}		
	\hfill	
	\begin{subfigure}[t]{0.48\textwidth}
		\centering
		\includegraphics[width=.92\textwidth,page=22]{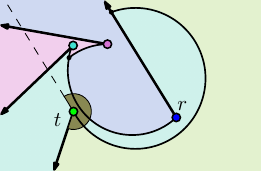}
		\caption{
			$\minangle$ is realized on ray $r$ at infinity, by ray $t$.
		}
		\label{fig:rvdBrocardPlane2}
	\end{subfigure}
	\caption{
		Two examples of the Brocard angle on a set $\cR$ of 4 rays in $\R^2$.  
	}
	\label{fig:rvdBrocardPlane}
\end{figure} 

Let $\minpoint \in \R^2$ be a point that realizes $\minangle$.
Conceptually, $\minpoint$ is the last point to be illuminated, assuming that all $n$ floodlights start with aperture $\alpha=0$ and simultaneously increase their apertures until the entire domain gets illuminated.
Although $\minpoint$ need not be unique, we show that it lies on $\rvd(\cR)$.

\begin{proposition}
	\label{prop:rvdRealizedVertex}
	The Brocard angle of a set $\cR$ of rays is realized at a vertex of $\rvd(\cR)$, or at a point at infinity along a ray in $\cR$.
\end{proposition}
\begin{proof}
    We first show that $\minpoint$ lies on $\rvd(\cR)$.
    Suppose, for the sake of contradiction, that $\minpoint$ does not lie on $\rvd(\cR)$, but instead, it lies inside the Voronoi region of a ray $r$.
    Then, we can always find a point with larger angular distance from $r$
    by simply moving in counterclockwise direction on the circle with center $\apex(r)$ and radius $d(\apex(r),\minpoint)$, deriving a contradiction;
    see for example the points $y$ and $y'$ in \cref{fig:rvdBrocardPlane1}.

    As pointed out in \cref{remark:distance}, the distance along a circular edge is monotone. Thus, the distance at one of the endpoints of the edge is at least as big as the distance at any point in the interior of the edge.
    This argument also holds for the distances along ray edges.
    Hence, a point with maximum distance $\minpoint$ is either a vertex of $\rvd(\cR)$ or a point at infinity on a ray of $\cR$, concluding the proof.
    Refer to \cref{fig:rvdBrocardPlane} for an illustration of the two cases.
\end{proof}

The above implies that we can find $\minpoint$, and hence $\minangle$, by first constructing $\rvd(\cR)$ in $O(n^{2+\epsilon})$ time 
and then traversing the diagram to find the vertex of maximum distance to its nearest neighbors.
$\rvd(\cR)$ is a plane graph, so it can be traversed in time linear in its size using standard methods.
This results in the following. 

\begin{theorem}
\label{thm:rvdFindBrocardPlane}
	The Brocard angle of a set $\cR$ of $n$ rays can be found in $O(n^{2+\epsilon})$ time. 
\end{theorem}

We conclude this section by giving tight bounds on the value of the Brocard angle.

\begin{proposition}
    \label{lem:boundAnglePlane}
    Given a set $\cR$ of $n$ rays, the range of values of the Brocard angle is $[2\pi/n,2\pi]$.
\end{proposition}
\begin{proof}
	For the upper bound consider a set $\cR$ of $n$ parallel rays:
	let $r_i$ have $p(r_i) = (i,0)$ and $\dir(r_i)=(1,0)$ for $i \in \{0,\dots,n-1\}$; see the example in \cref{fig:rvdBrocardAnglePlaneUPPER}.
	Observe that the last point to be illuminated is the point on $r_0$ at infinity, i.e., point $(0,+\infty)$, which will be illuminated by $r_{n-1}$ when $\alpha$ reaches $2\pi$;
	hence the upper bound follows.

\begin{figure}[b]
    \centering
    \begin{subfigure}[t]{0.48\textwidth}
        \centering
        \includegraphics[width=0.92\textwidth,page=2]{rvdBrocardBounds.pdf}
        \caption{
            $\minangle = 2\pi$.
        }
        \label{fig:rvdBrocardAnglePlaneUPPER}
    \end{subfigure}
    \hfill
    \begin{subfigure}[t]{0.48\textwidth}
        \centering
        \includegraphics[width=.92\textwidth,page=1]{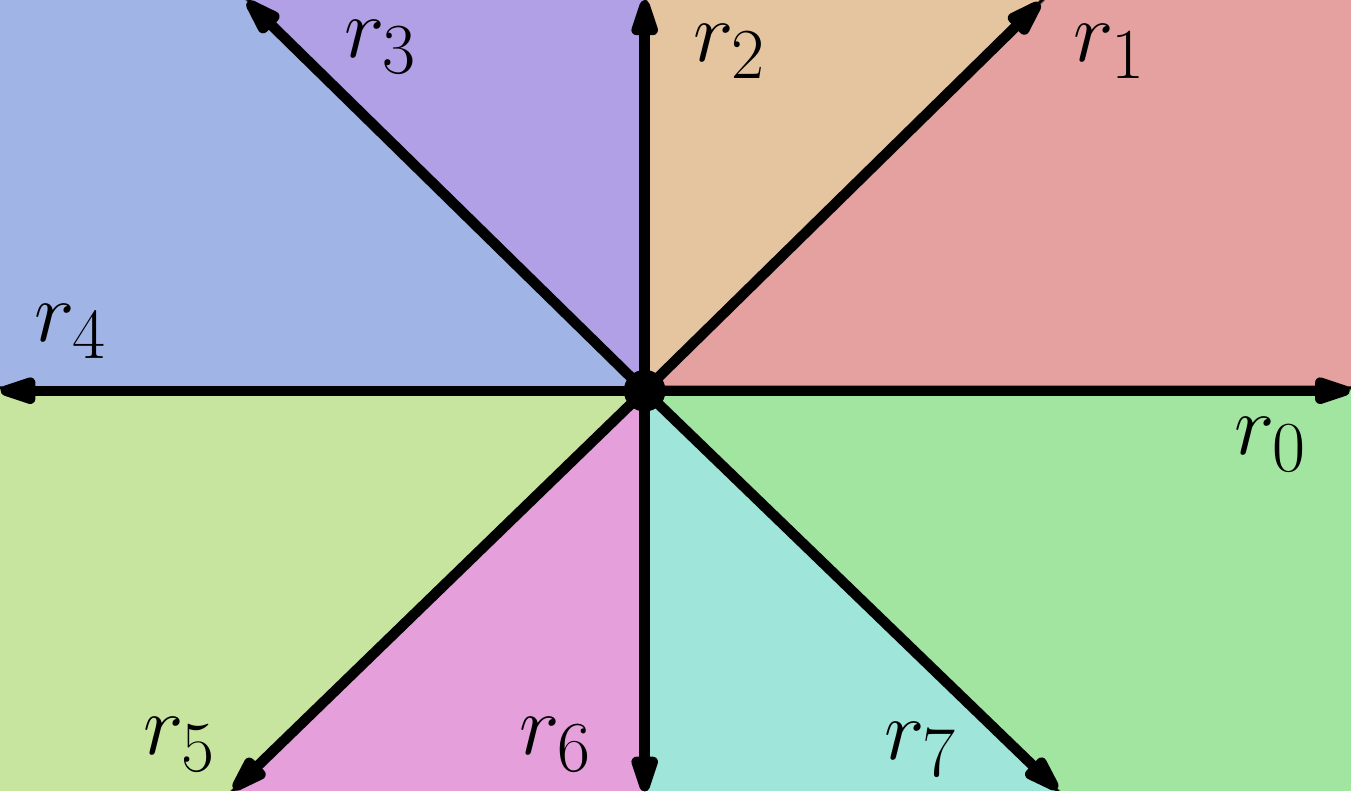}
        \caption{
            $\minangle = 2\pi/n$.
        }
    \label{fig:rvdBrocardAnglePlaneLOWER}
    \end{subfigure}  
    \caption{
			Sets of $8$ rays realizing the bounds of the Brocard angle in $\R^2$.    
	}
  \label{fig:rvdBrocardAnglePlane}
\end{figure}
 
	For the lower bound, consider that in order to illuminate the entire $\R^2$, all the points \emph{at infinity} should also be illuminated.
	To illuminate such points, the sum of the angles of all rays, should be at least $2\pi$.
	Hence, in the best case, a point at infinity is seen by exactly one ray, and a $2\pi/n$ lower bound follows.
	A construction achieving the $2\pi/n$ bound is the following.
	Let $\cR$ be a set of $n$ rays having apex at $(0,0)$ and with the property that any two consecutive rays have an angular difference of $2\pi/n$;
	see the example in \cref{fig:rvdBrocardAnglePlaneLOWER}.
	The last points to be illuminated will be all the points on the right side of each ray $r_i$.
	These points are illuminated simultaneously by $r_{i-1}$ when $\alpha$ reaches $2\pi/n$.
	Further, the above construction can be easily adapted to attain any value in $(2\pi/n, 2\pi)$, by expanding a wedge formed by two consecutive rays and shrinking all the others accordingly.
\end{proof}

\section{Rotating Rays Voronoi diagram in a convex polygon}
\label{sec:polygon}

In this section we describe a linear time algorithm to construct the Rotating Rays Voronoi Diagram restricted to the interior of a convex polygonal region.
We also show how to use this algorithm to compute the Brocard angle of a convex polygon in optimal linear time.

Throughout this section we use the following notation.
We denote by $\Pol$ a convex polygon with $n$ vertices, and by $v_1,\ldots,v_n$ the vertices of $\Pol$ labeled by appearance while traversing the boundary of $\Pol$ in counterclockwise direction.
For the sake of simplicity, we assume that arithmetic operations on indices are taken modulo $n$.
We denote with $\cR_{\Pol} = \{ r_1,\ldots,r_n \}$ the set of $n$ rays such that the ray $r_i$ leaves the vertex $v_i$ and passes through the vertex $v_{i+1}$, see \cref{fig:convex_instance_intro}.
We finally denote with $\pvd(\cR_{\Pol}):= \rvd(\cR_{\Pol}) \cap \Pol$ the Rotating Rays Voronoi Diagram of $\cR_{\Pol}$ restricted to the interior of $\Pol$.
Examples of this diagram are shown in \Cref{fig:the_problem2,fig:convex_parallel}.

\begin{figure}[t]
  \centering
  \begin{minipage}[t]{0.37\textwidth}
    \centering
    \includegraphics[width=\textwidth, page=1]{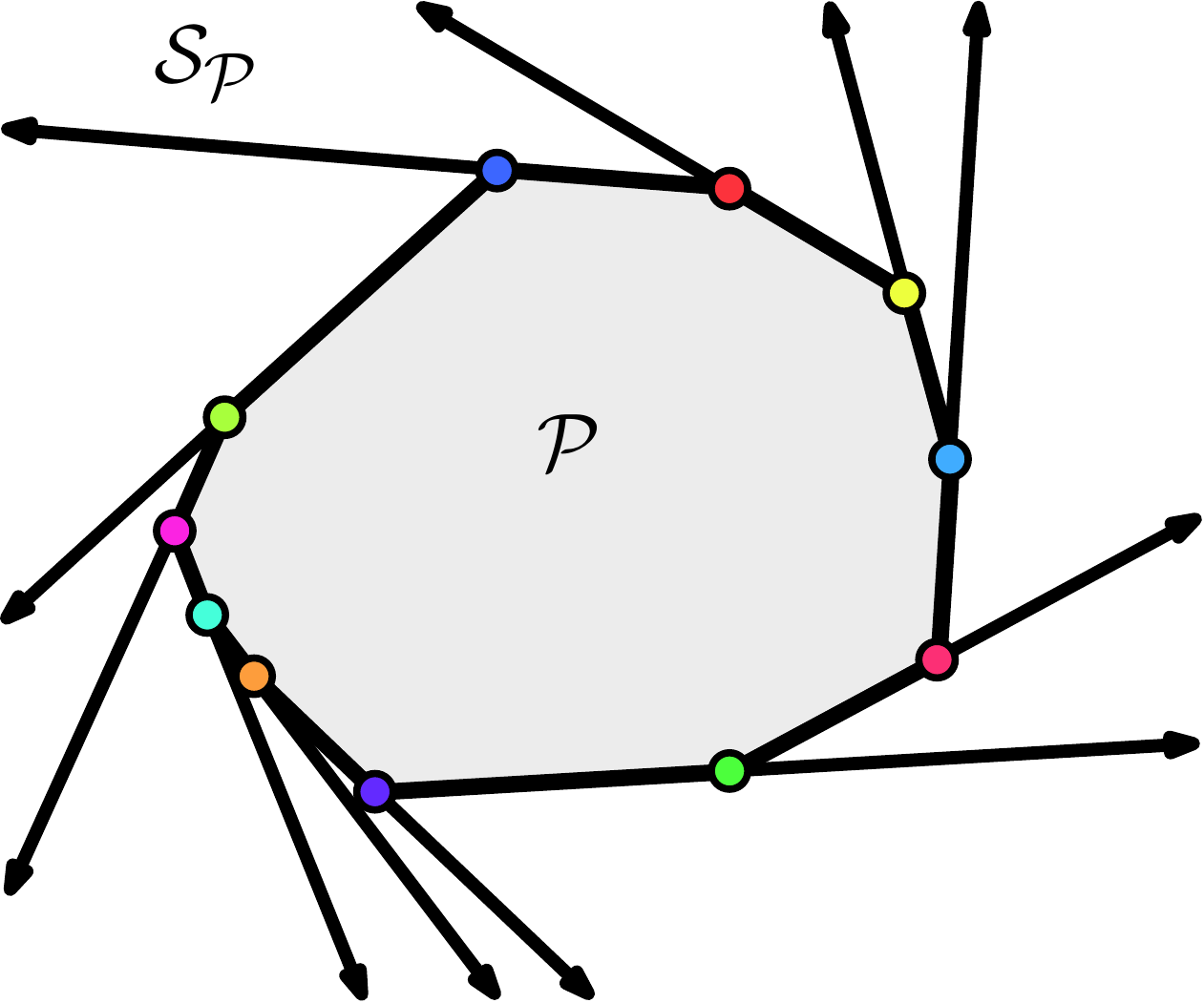}%
    \caption{A convex polygon $\Pol$ and the set of rays $\cR_\Pol$.}
    \label{fig:convex_instance_intro}
  \end{minipage}
  \hfill
  \begin{minipage}[t]{0.6\textwidth}
    \centering
    \includegraphics[trim=0 30 0 30, clip, width=0.8\textwidth, page=1]{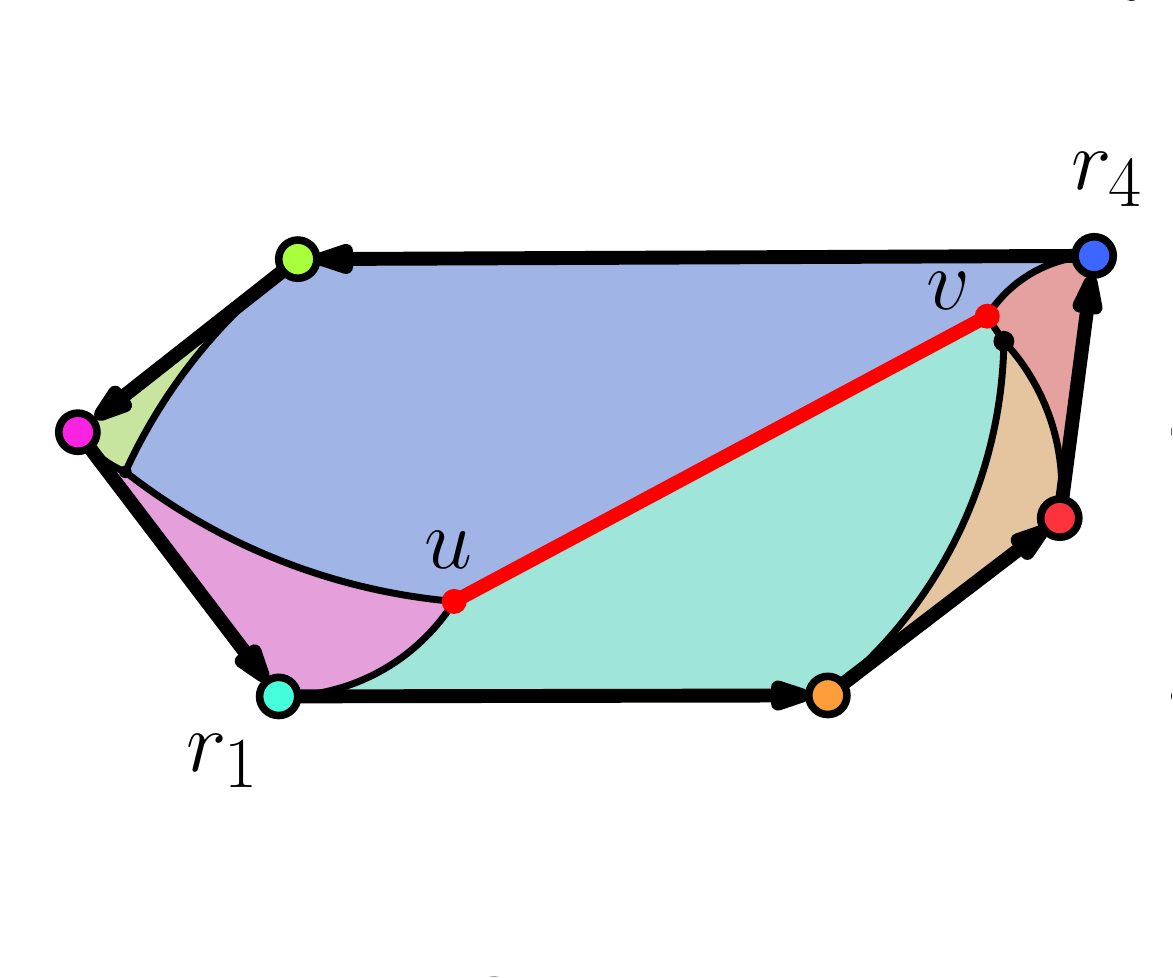}%
    \caption{A convex polygon $\Pol$ with two anti-parallel edges $r_1,r_4$. All points on edge $uv$ of $\pvd(\cR_\Pol)$ realize $\minangle$.
    }
    \label{fig:convex_parallel}
  \end{minipage}
\end{figure}

This section is organized as follows.
In \cref{subsec:pvdProperties} we describe a set of basic properties of $\pvd(\cR_{\Pol})$.
In \cref{subsec:pvdNlogn} we show that $\pvd(\cR_{\Pol})$ can be computed in $O(n\log n)$ time and $O(n)$ space.
Our main result is presented in \cref{subsec:pvdLinear,subsec:pvd4diagrams,subsec:pvdMerging} showing that the diagram $\pvd(\cR_{\Pol})$ can be computed in $\Theta(n)$ time and $O(n)$ space.
Finally, in \cref{subsec:pvdBrocard} we discuss the implications of this result to the computation of the Brocard angle of $\Pol$, and to related illumination problems.

\subsection{Properties of the diagram}
\label{subsec:pvdProperties}

For the sake of simplicity, throughout this section we make the following assumptions.
First, no three vertices of $\Pol$ are collinear.
Second, no point in the boundary or the interior of $\Pol$ is equidistant to four rays of $\cR_{\Pol}$; 
this implies that all vertices of $\pvd(\cR_{\Pol})$ are incident to at most three edges.
Finally, we assume that there are no parallel edges in $\Pol$, i.e., there are no anti-parallel rays in $\cR_\Pol$.
The last assumption guarantees that the Brocard angle is realized at a unique point, which is a vertex of $\rvd(\cR_\Pol)$.
Recall from \cref{sec:preliminaries} that the bisector of two anti-parallel edges $r$ and $s$ contains the line segment connecting $\apex(r)$ and $\apex(s)$.
Thus, if there are anti-parallel rays, the Brocard angle may be realized on any point of a Voronoi edge, which is part of the bisector of two anti-parallel edges; see, e.g., the edge $\ov{uv}$ in \cref{fig:convex_parallel}.

In the following, we present some useful properties of $\pvd(\cR_\Pol)$.

\subparagraph{Disk diagram.}
We first define an auxiliary Voronoi diagram, the \emph{Disk Diagram ($\dd$)}, whose system of bisectors consists solely of the bisecting circles.
This diagram is simpler to study and coincides with the rotating rays Voronoi diagram within the convex polygon $\Pol$.

Formally, given two rays $r$ and $s$, we define the \emph{$\dd$ bisector} of $r$ and $s$ to be the entire bisecting circle $\cb(r,s)$.
The \emph{$\dd$ dominance region} of $r$ over $s$, denoted by $\dr_D(r,s)$, is either the interior or the exterior of this circle, depending on the angular difference of the two rays.
If $\diff(r,s) < \pi$, then $\dr_D(s,r)$ is the interior of $\cb(r,s)$ and $\dr_D(r,s)$
the exterior, and the other way round if $\diff(r,s) \geq \pi$.

The DD region of a ray $r \in \cR_\Pol$ is $\dreg(r) := \bigcap_{s \in \cR_\Pol \setminus \{r\}} \dr_D(r,s)$.
The disk diagram is $\dd(\cR_\Pol)= \R^2\setminus \bigcup_{r \in \cR_\Pol}  \dreg(r)$;
see \cref{fig:diskDiagram2}. 

The disk diagram does not necessarily cover $\R^2$, for $n\geq 3$. 
This is because there are areas having a cyclic dominance relation among some sites; see for example the white/uncolored region in \cref{fig:diskDiagram2}.

\begin{figure}[t]
  \centering
    \begin{subfigure}[t]{0.48\textwidth}
    \includegraphics[width=\textwidth,page=1]{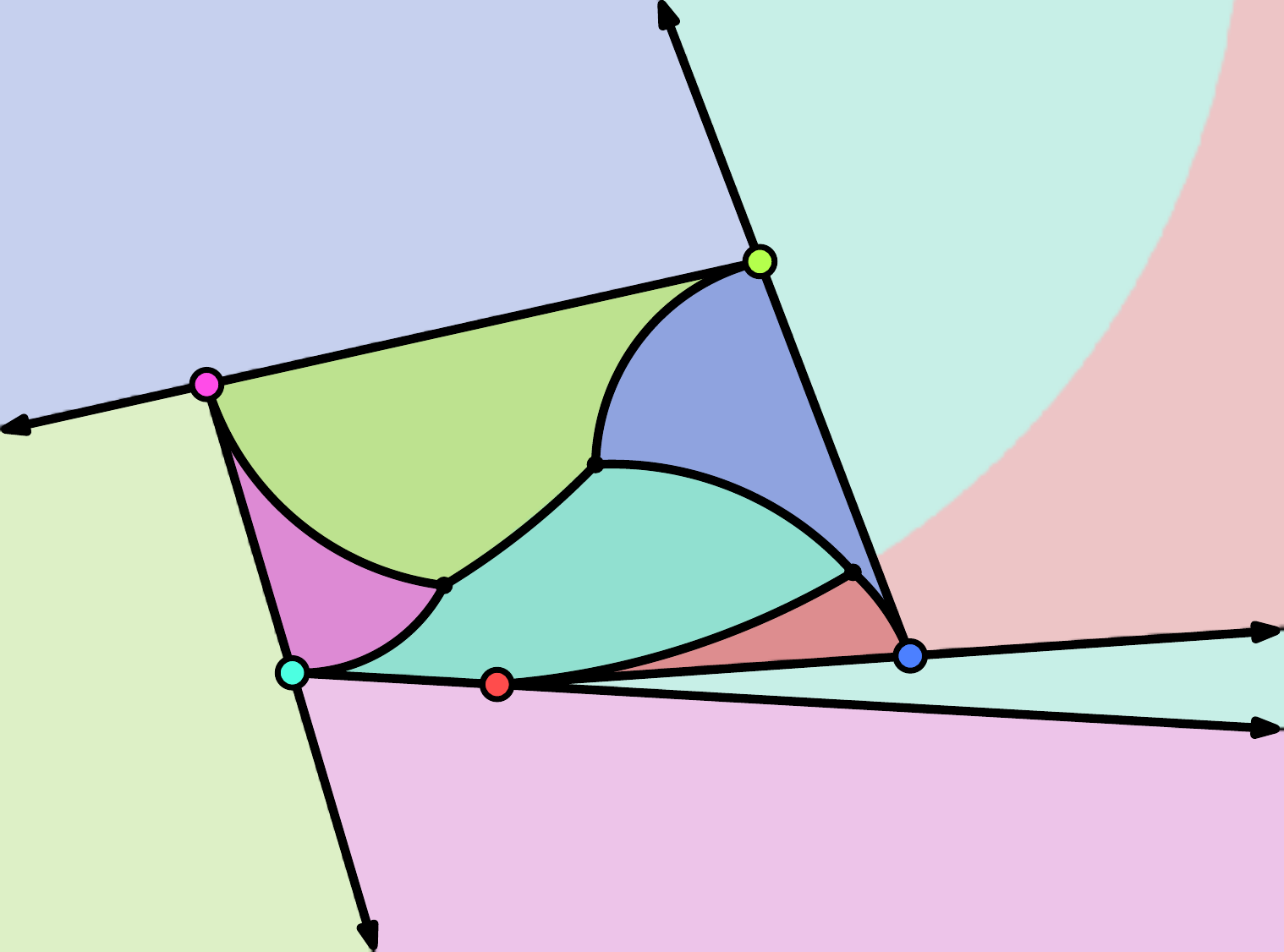}
    \caption{
    The diagram $\pvd(\cR_{\Pol})$ shown highlighted.
    The diagram $\rvd(\cR_{\Pol})$ in $\mathbb{R}^2 \setminus \Pol$ is shown faded.
    }
    \label{fig:diskDiagram1}
    \end{subfigure}
    \hfill
    \begin{subfigure}[t]{0.48\textwidth}
    \includegraphics[width=\textwidth,page=2]{rvdConvexNlogn.pdf}
    \caption{The diagram $\dd(\cR_{\Pol})$.
    The diagram inside $\Pol$ is shown highlighted.
    }  \label{fig:diskDiagram2}
    \end{subfigure}
  \caption{
    A convex polygon $\Pol$ with five vertices together with $\pvd(\cR_{\Pol})$ and $\dd(\cR_{\Pol})$.
  }
  \label{fig:diskDiagram}
\end{figure}

With this definition, every point in the neighborhood of a circular edge in $\rvd(\cR_\Pol)$ is associated to the same ray in $\rvd(\cR_\Pol)$ and in $\dd(\cR_\Pol)$.
Since inside $\Pol$ the rotating rays Voronoi diagram consists only of circular edges, it follows that $\rvd(\cR_\Pol)$ and $\dd(\cR_\Pol)$ are exactly the same in $\Pol$; see \cref{fig:diskDiagram}.

\begin{lemma}
  \label{lem:rvdConvexDiskDiagram}
  Each region $\dreg(r)$ of $\dd(\cR)$ is connected and contains $\apex(r)$ on its boundary.
\end{lemma}
\begin{proof}
	By definition, the region $\dreg(r)$ is formed by 
	the intersection of $n-1$ disks or their counterparts.
	The boundary of each of these disks is $\cb(r,s)$ for some $s$, and since by definition the apex $\apex(r)$ lies on $\cb(r,s)$, it also lies on the boundary of the intersection of all these $\cb(r,s)$ disks.
	
	To see that the region $\dreg(r)$ is connected, we perform an inversion of the plane using $\apex(r)$ as the inversion center, and a circle of arbitrary radius as the inversion circle.
	This inversion maps circles passing through the inversion center to lines passing through the inversion center, so each dominance region $\dr_D(r,s)$ maps to a halfplane.
	The intersection of halfplanes is connected, and since the inversion preserves connectivity, then region $\dreg(r)$ is also connected.
\end{proof}

\begin{corollary}
 \label{cor:pvdComplexity}
 The truncated portion of the diagram $\dd(\cR_\Pol)$ within $\Pol$ has a tree structure and is of $\Theta(n)$ complexity. It coincides with $\pvd(\cR_\Pol)$.
\end{corollary}

We now turn our attention back to $\pvd(\cR_\Pol)$; we will use  the disk diagram in \cref{subsec:pvdMerging}.

\begin{lemma}\label{lem:rvdConvexChain}
In $\pvd(\cR_{\Pol})$, the Voronoi region of a site $r_i$ is connected, incident to the polygon edge $p(r_i)p(r_{i+1})$, and has the following form: It consists of a single face incident to the polygon edge $p(r_i)p(r_{i+1})$. The distance on the boundary is monotonically increasing from $r_i$, and from $r_{i+1}$, towards a global maximum at $r_i^*$.See \cref{fig:properties}.
\end{lemma}

\begin{proof}
    Consider the sequence of sites whose faces are adjacent to the face of $r_i$ in counterclockwise order. We show that: \textbf{(i)} this sequence is actually a sub-sequence of $(r_{i+1}, r_{i+2},..., \allowbreak r_n, r_1,...,r_{i-1})$,
    and \textbf{(ii)} the distance along the sequence is monotonically increasing towards a point realizing the maximum $r_i^*$.
    Refer to \cref{fig:properties} (where $r_i = r_1$).
    
    \textbf{(i)}
    Let $r_j$ be a ray such that $\rreg(r_i)$ is adjacent to $\rreg(r_j)$.
    By \cref{lem:rvdConvexDiskDiagram}, each region is connected and incident to its corresponding ray, so the union of $\rreg(r_i)$ and $\rreg(r_j)$ \emph{splits} the convex polygon into two simply connected components
    (see $\rreg(r_1) \cup \rreg(r_5)$ in \cref{fig:properties}).
    Given a second ray $r_k$ whose region is adjacent to the region $\rreg(r_i)$, it follows that both $r_k$ and the edge between $\rreg(r_i)$ and $\rreg(r_k)$ have to be in the same connected component, as also $\rreg(r_k)$ is connected (see $\rreg(r_3)$ in \cref{fig:properties}). Thus, the order of $r_j$ and $r_k$ along the boundary of the polygon and the face of $r_i$ is the same.
    
    \textbf{(ii)}
    We partition the set $\cR_\Pol \setminus p_i$ in two depending on the angular difference with $p_i$. 
    The set $\cR_{\geq \pi}$ of rays with angular difference $\diff(r_j,r_i) \geq \pi$, and the set $\cR_{<\pi}$ of rays with angular difference  $\diff(r_j,r_i) < \pi$
    (see the dashed curves in \cref{fig:properties})
    Along the chain of rays $\cR_{<\pi}$ (resp. $\cR_{\geq \pi}$) the distance is increasing from $r_{i+1}$ (resp. $r_{i}$) towards the point realizing $r_i^*$.

\begin{figure}[t]  
    \centering
        \includegraphics[width=0.6\textwidth, page=3]{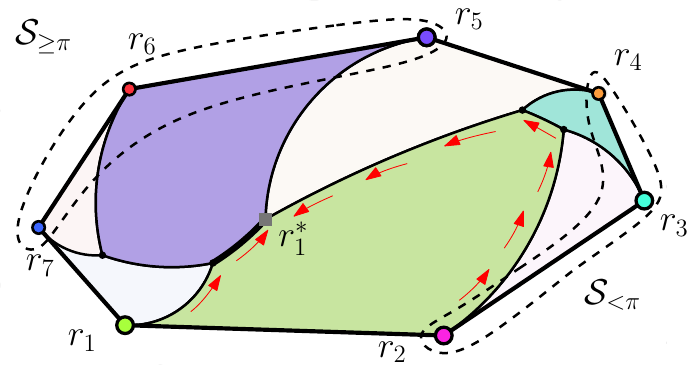}
    \caption{
    Illustration of the properties of a region $\rreg(r_1)$.
    The distance along the boundary $\partial\rreg(r_1)$ is increasing towards the maximum $r_1^*$.
    The sites in $\cR_\Pol \setminus r_1$ are split in two sets $\{r_2,r_3,r_4\}$ and $\{r_5,r_6,r_7\}$.
    The regions of $r_1$ and $r_5$ split $\Pol \setminus \{\rreg(r_5)\cup\rreg(r_1)\}$ in two connected components.
    }
    \label{fig:properties}
\end{figure}

    We give an inductive argument for the monotonicity property along chain $\cR_{<\pi}$, starting initially only with the Voronoi diagram of the two rays $r_i$ and $r_{i+1}$, and then incrementally adding more rays of the set $\cR_{<\pi}$ in counterclockwise order.
    The base case follows directly from the properties of the bisectors, see \cref{remark:distance}.
    Suppose we are now adding site $r_k$ to $\rvd(\{r_i,\dots,r_{j-1}\})$. Because of property \textbf{(i)}, if there is an edge between $r_i$ and $r_j$, then it is incident to $\apex(r_i)$, i.e., it is the last one along the chain of edges of face $\rreg(r_i)$. 
    Let $v$ be the other endpoint of the edge between $r_i$ and $r_j$.
    Since $r_j \in \cR_{<\pi}$, the distance along the edge $r_i$ and $r_j$ is monotone increasing, from $v$ to $\apex(r_i)$. 
    Further, by the induction hypothesis, the distance along the chain of edges between $r_i$ and all sites bounding $\rreg(r_i)$ in $\rvd(\{r_i,\dots,r_{j-1}\})$ is monotonically increasing in counterclockwise order, from $r_{i+1}$ to $v$.
    The proof for $\cR_{\geq \pi}$ is analogous, but instead, the distance increases in clockwise order.
\end{proof}

\begin{corollary}\label{lem:2IncomingEdges}
    For any vertex $u \in \pvd(\Pol)$, at least two incident edges have a distance increasing towards $u$.
\end{corollary}

\begin{proof}
    Assume for the sake of contradiction, that there is a vertex $u$ with two incident edges having distance decreasing towards $u$.
    These two edges are part of a chain of edges bounding a region $\rreg(r_i)$.
    But this contradicts \cref{lem:rvdConvexChain}.
\end{proof}

\begin{lemma} \label{lem:PointsOfSameDistance}
  Given an angle $c$, the set of all points in the interior of $\Pol$, which have at least distance $c$ to their nearest site, form a convex polygon $B$.
\end{lemma}
\begin{proof}
Observe that the set of all points at distance $c$ from a ray of $\cR_\Pol$ is a half-line.
Thus, $B$ is a convex polygon, since it is the intersection of the halfplanes defined by all the rays in $\cR_\Pol$, after being rotated by $c$.
\end{proof}

\subsection{A simple \texorpdfstring{$O(n \log n)$}{O(nlogn)}-time algorithm}
\label{subsec:pvdNlogn}

We first describe a simple $O(n\log n)$-time algorithm to construct $\pvd(\cR_\Pol)$ employing a so-called \emph{``collapse''} strategy: starting at the boundary of the domain (where all the points have distance zero to its nearest site), the algorithm gradually constructs the diagram adding edges and vertices of increasing 
distance until the vertex of maximum distance is reached.
Other examples of algorithms employing a similar strategy
include the farthest Voronoi diagrams of points \cite{skyum1991} and of line segments \cite{aurenhammer2006}.
We remark that the $O(n\log n)$-time algorithm of \cite{alegria2017} finds the Brocard angle in a similar manner, without constructing the $\pvd(\cR_\Pol)$.

We give a high level description of the algorithm; refer to \cref{fig:algorithmnlogn} for an illustration.
The algorithm starts at the vertices of $\Pol$, which are all starting points of edges of $\pvd(\cR_\Pol)$.
For every pair of edges that are consecutive in circular order, their next intersection point is computed, if one exists.
Out of these intersection points, the one with minimum distance to its nearest site is the next vertex of the diagram.

\begin{figure}[!b]
    \centering
    \begin{subfigure}[t]{0.48\textwidth}
    \centering
        \includegraphics[width=.95\textwidth,page=4]{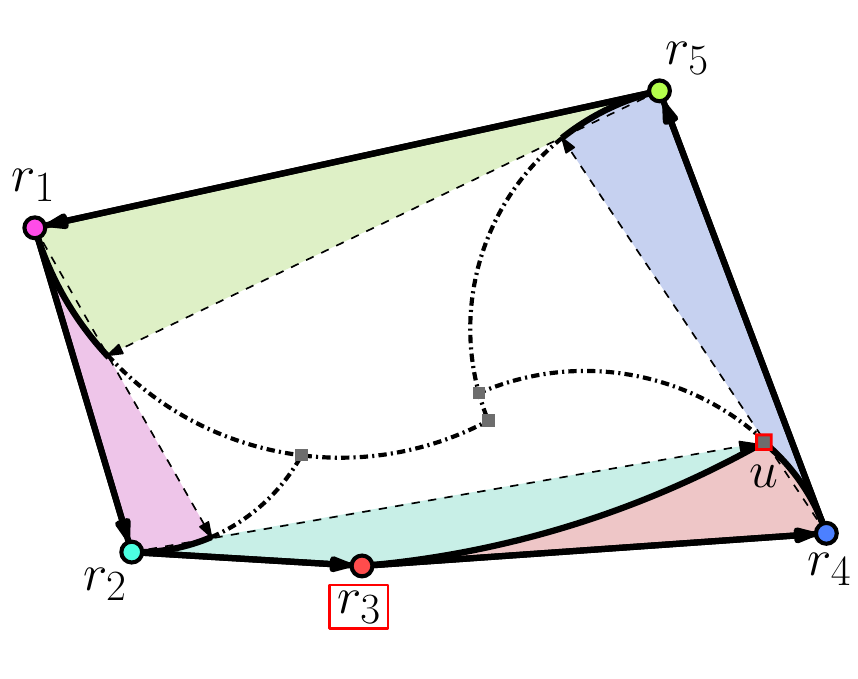}
    \caption{
        First event: Vertex $u$ is induced by the rays $(r_2,r_3,r_4)$, and the region of $r_3$ ``collapses''.
    }
    \end{subfigure}
    \hfill
    \begin{subfigure}[t]{0.48\textwidth}
        \centering
        \includegraphics[width=.95\textwidth,page=5]{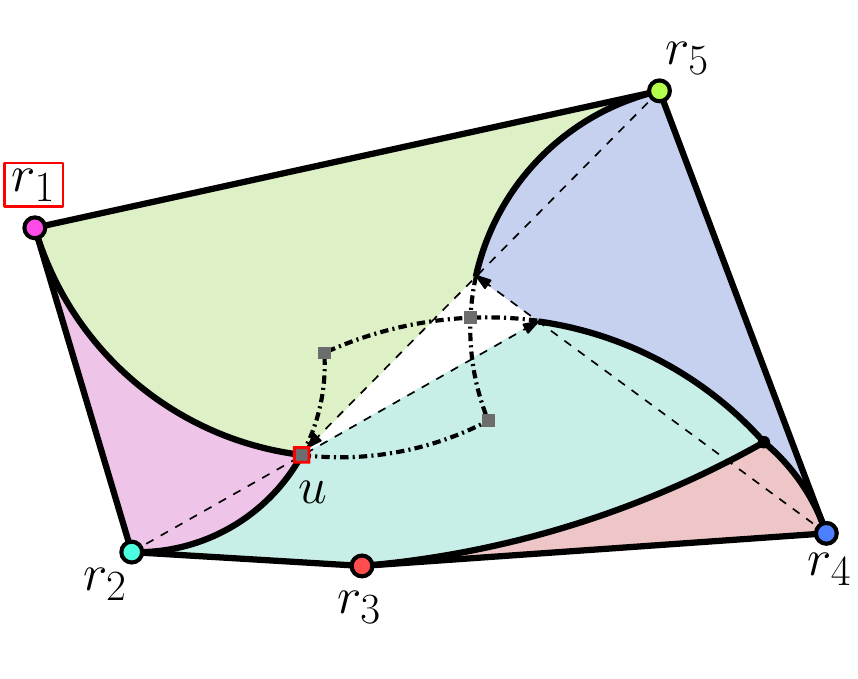}
    \caption{
        Second event: Vertex $u$ is induced by the rays $(r_5,r_1,r_2)$, and the region of $r_1$ ``collapses''.
    }    
    \end{subfigure}
  \caption{
   An illustration of the first two events of the $O(n\log n)$-time algorithm. 
   At each event all the candidate vertices are illustrated; 
   In the third event (not illustrated) there is only one candidate vertex induced by the rays $r_2,r_4,r_5$.
  }
  \label{fig:algorithmnlogn}
\end{figure}

This \emph{collapse} event is processed by constructing the vertex and the edges leading to this vertex, then removing the edges from further consideration and starting a new edge.
At the constructed vertex a face \emph{collapses}, since it is fully constructed and will not be considered again. The new edge is part of the bisector of the two faces neighboring the collapsed face.
This procedure, of computing and processing new \emph{collapse}  events, is repeated until all the remaining edges intersect in a single point.
This last point is the vertex of $\pvd(\cR_\Pol)$ with maximum distance, and realizes the Brocard angle of $\Pol$.

\subparagraph{Correctness of the algorithm.} The algorithm constructs the features of the diagram (vertices and edges) in increasing distance to their nearest site. First note, that because of \cref{lem:PointsOfSameDistance}, the set of points with same distance to their closest site form a single cycle (see the dashed polygons in \cref{fig:algorithmnlogn}).
So to keep track of the edges which are \emph{``active'' candidates} to be added to the diagram it suffices to use a circular list.

We argue that the algorithm correctly finds the next vertex, assuming that all vertices with smaller distance have already been constructed. 
By \cref{lem:2IncomingEdges}, every vertex has at least two edges with increasing distance towards it, thus, the next vertex is the intersection of a pair of edges currently in the circular list. Such a pair of edges needs to be consecutive, as otherwise, the 
cyclicity of \cref{lem:PointsOfSameDistance} would be violated. Since we are constructing the diagram in increasing distance, the algorithm picks the candidate vertex with smallest distance.

Finally, we argue that the algorithm does not miss any candidate edge.
This is because the algorithm considers a new candidate edge at each vertex event, and edges can only start at vertices, since the distance function does not exhibit local minima along an edge; see the monotonicity property of \cref{remark:distance}.

\begin{proposition}
    The ``collapse'' algorithm 
    constructs $\pvd(\cR_\Pol)$ in $O(n\log n)$ time.
\end{proposition}

\begin{proof}
The time complexity analysis is straightforward.
The algorithm takes $O(n\log n)$ time to sort the first $n$ events. 
Then, there  are $n-2$ events and each event takes $O(\log n)$ time, if we use a min priority queue.
Thus, the algorithm takes $O(n\log n)$ time overall.
\end{proof}

\subsection{An optimal \texorpdfstring{$\Theta(n)$}{O(n)}-time algorithm}
\label{subsec:pvdLinear}

We now describe how $\pvd(\cR_\Pol)$ can be constructed in optimal $\Theta(n)$-time.
We have already shown that $\pvd(\cR_\Pol)$ has a tree structure and each Voronoi region is connected.
Despite its simple structure, however, $\pvd(\cR_\Pol)$ is not an instance of \emph{abstract Voronoi diagrams}~\cite{klein1989,klein1994} as we will see in the sequel.  
Thus, we can not use the available machinery under this framework directly.
Instead, we split the problem into sub-problems, where each sub-problem falls under the abstract Voronoi diagram framework.

Our algorithm can be briefly described as follows, see also the pseudocode in \cref{algorithm2}.
In a first step we partition $\cR_{\Pol}$ into four sets $\cR_{N}, \cR_{W}, \cR_{S}$ and $\cR_{E}$ of consecutive rays, depending on whether a ray faces \emph{north}, \emph{west}, \emph{south} or \emph{east} respectively; see an example in \cref{fig:convex_instance_partitioning}.

\begin{figure}[b]
    \centering
    \begin{subfigure}[t]{0.48\textwidth}
        \centering
        \includegraphics[width=0.75\textwidth, page=3]{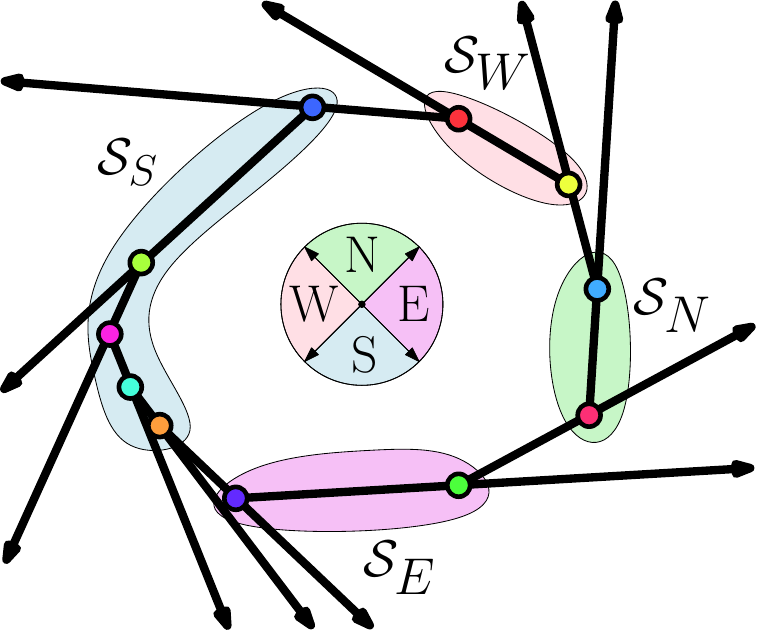}%
        \caption{
            The sets of rays $\cR_{N}, \cR_{W}, \cR_{S}$, $\cR_{E}$.
        }
        \label{fig:convex_instance_partitioning}
    \end{subfigure}
    \hfill
    \begin{subfigure}[t]{0.48\textwidth}
        \centering
        \includegraphics[width=0.75\textwidth, page=4]{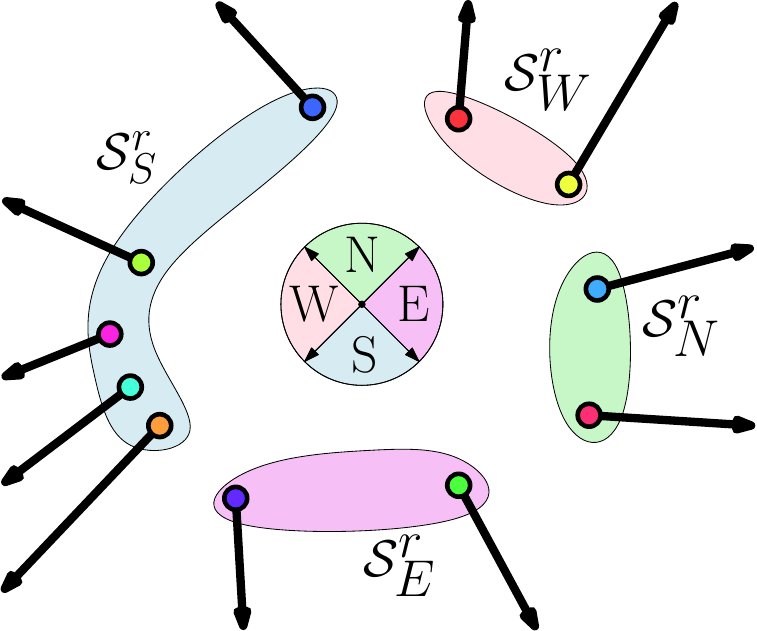}%
        \caption{
            The sets of rays $\cR_{N}^r, \cR_{W}^r, \cR_{S}^r$, $\cR_{E}^r$.
        }
        \label{fig:convex_instance_rotation}
    \end{subfigure}
    \caption{
        The partitioning of the set of rays in $\cR_{\Pol}$ before and after rotation. 
    }
  \label{fig:convex_instance_4sets}
\end{figure}

In a second step we transform each set $\cR_d$, $d \in \{\textup{N,W,S,E}\}$ into a set $\cR^r_d$, where each ray in $\cR_d$
is rotated clockwise by an angle of $\pi/2$;
see \cref{fig:convex_instance_rotation}.
We then construct each diagram $\rvd(\cR^r_d)$ independently as a special instance of abstract Voronoi diagrams; see \cref{fig:convex_instance1}.
Finally, we merge the four diagrams and  obtain $\pvd(\cR_{\Pol})$.
The merging is done in two phases; 
see \cref{fig:convex_instance_merge_1} and \cref{fig:convex_instance_merge_2}.

\begin{algorithm}[t]
	\SetAlgoVlined 
	\SetAlgoNoEnd 
	\LinesNumbered
	\SetKwInOut{Input}{Input} \SetKwInOut{Output}{Output}
	\SetKw{Merge}{Merge}
	\SetKw{Split}{Split}
	\SetKw{Construct}{Construct}
	\SetKw{Rotate}{Rotate}
	\SetKw{Update}{Update}
	\SetKwComment{Comment}{//}{}
	\Input{A convex polygon $\Pol$ with $n\geq 3$ vertices.}
	\Output{The diagram $\pvd(\cR_\Pol)$.}

    $\{\cR_N,\cR_W,\cR_S,\cR_E\} \gets$ \Split $\cR_{\Pol}$ \;
	\For{\textup{each} $d \in \{N,W,S,E\}$}
	{
	    $\cR^r_d \gets$ \Rotate $\cR_d$ \;
	    \Construct $\rvd(\cR^r_d)$ \;
	}
	
    $\rvd(\cR^r_W \cup \cR^r_S) \gets$ \Merge $\rvd(\cR^r_W) \ \text{and} \ \rvd(\cR^r_S)$ \;
    $\rvd(\cR^r_N \cup \cR^r_E) \gets$ \Merge $\rvd(\cR^r_N) \ \text{and}  \ \rvd(\cR^r_E)$ \;
    $\pvd(\cR_\Pol) \gets$ \Merge $\rvd(\cR^r_W \cup \cR^r_S) \ \text{and}  \ \rvd(\cR^r_N \cup \cR^r_E)$ \;

	\Return{$\pvd(\cR_\Pol)$} \; 
	\caption{$\Theta(n)$-time algorithm to construct $\pvd(\cR_\Pol)$.}%
	\label[algo]{algorithm2}	
\end{algorithm}

We describe in detail the construction of the four diagrams in \cref{subsec:pvd4diagrams} and the merging phase in \cref{subsec:pvdMerging}.
Consequently we 
derive the following theorem.

\begin{restatable}{theorem}{convexAlgo}
\label{thm:convex_algorithm}
  Given a convex polygon $\Pol$, we can construct $\pvd(\cR_{\Pol})$ in deterministic optimal $\Theta(n)$~time.
\end{restatable}

\subsection{\texorpdfstring{$\Theta(n)$}{O(n)}-time algorithm: constructing the four diagrams}
\label{subsec:pvd4diagrams}

We use the framework of \emph{abstract Voronoi diagrams}~\cite{klein1989,klein2009}.
To comply with this framework, a system of angular bisectors must satisfy the following three \emph{axioms}:

\begin{itemize}[noitemsep]
\item \textbf{(A1)}
  The bisector $\rbis(r,s)$, $\forall r,s \in \cR$, is an unbounded simple curve, homeomorphic to a line. 
\item \textbf{(A2)}
  The region $\rreg(r)$ in $\rvd(\cR')$, $\forall \cR' \subseteq \cR$ and $\forall r
  \in \cR'$, is connected.
\item \textbf{(A3)}
  The closure of the union of all regions in $\rvd(\cR')$, $\forall \cR' \subseteq
  \cR$, covers $\mathbb{R}^2$.
\end{itemize}

Observe that a subset $\cR_d$ (and hence the set $\cR_{\Pol}$) need not satisfy axiom (A2); see the disconnected Voronoi regions in \cref{fig:convex_instance1before}.

\begin{figure}[b]
    \centering
    \begin{subfigure}[t]{0.48\textwidth}
        \centering
        \includegraphics[width=0.75\textwidth, page=7]{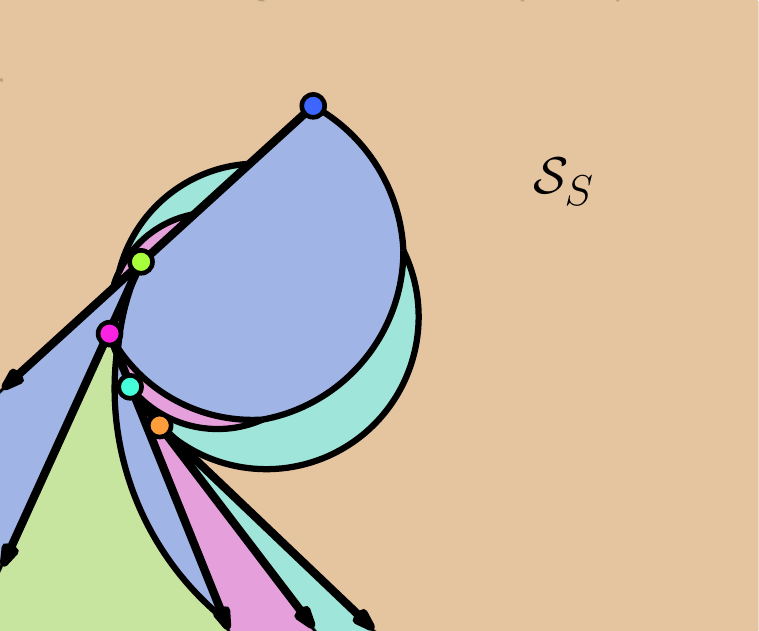}%
    \caption{
    Diagram $\rvd(\cR_{S})$.
    }
    \label{fig:convex_instance1before}
    \end{subfigure}
    \hfill
    \begin{subfigure}[t]{0.48\textwidth}
        \centering
        \includegraphics[width=0.75\textwidth, page=8]{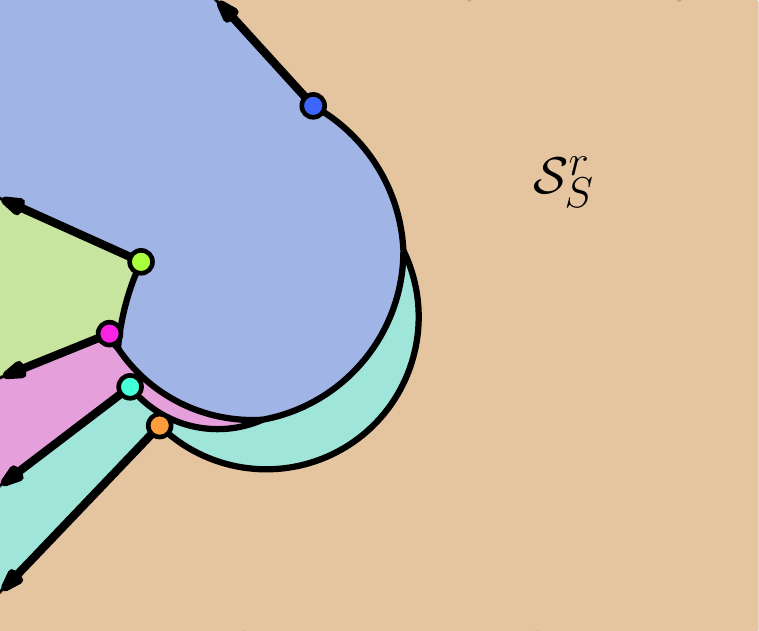}%
    \caption{
    Diagram $\rvd(\cR^r_{S})$.
    }
    \label{fig:convex_instance1after}
    \end{subfigure}
    \caption{
    Voronoi diagrams of a set $\cR_S$ and the set $\cR_S^r$ (after a clockwise rotation by $\pi/2$).
    }
  \label{fig:convex_instance1}
\end{figure}

Consider each transformed subset of rays $\cR^r_d$, $d \in \{\textup{N,W,S,E}\}$, where each ray 
is rotated clockwise by an angle of $\pi/2$.
We will show that each set $\cR^r_d$ satisfies the aforementioned axioms.
The intuition behind the clockwise rotation comes from the fact that only circular parts of bisectors appear in $\pvd(\cR_{\Pol})$, and  these bisecting circles remain the same under a uniform rotation.
An example of a diagram that satisfies axiom (A2)
after a clockwise $(\pi/2)$-rotation is shown in
\cref{fig:convex_instance1after}.

Still, note that a clockwise $(\pi/2)$-rotation by itself is not always sufficient for the entire set $\cR_{\Pol}$ to satisfy axiom (A2), and this is justified by the example in \cref{fig:noangle_example}:
given the set $\cR_\Pol$ of $5$ rays in \cref{fig:noangle_example0}, there exists a subset of $3$ rays which needs more rotation in order to have all regions connected, see $\rreg(s)$ in \cref{fig:noangle_example1}, and a subset of $3$ rays which needs less rotation, see $\rreg(r)$ in \cref{fig:noangle_example2}; hence, the reason to split $\cR_\Pol$ appropriately.

\begin{figure}[t]  
    \centering
    \begin{subfigure}[t]{0.83\textwidth}
        \centering
        \includegraphics[trim=0mm 42mm 0mm 15mm, clip,width=\textwidth,page=6]{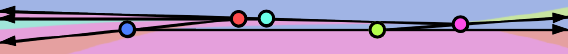}
        \caption{
            The Voronoi diagram of the complete set of rays without rotation (zoomed in).}
\label{fig:noangle_example0}
\end{subfigure}
\begin{subfigure}[t]{0.48\textwidth}
\centering
\includegraphics[width=.65\textwidth,page=3]{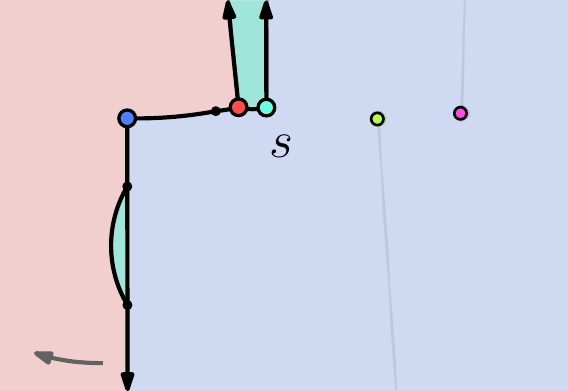}
\caption{A subset of $3$ rays clockwise $(\pi/2)$-rotated. $\rreg(s)$ has 2 faces; 
to become connected, more rotation is needed.} 
\label{fig:noangle_example1}
\end{subfigure}
    \hfill
    \begin{subfigure}[t]{0.48\textwidth}
        \centering
        \includegraphics[width=.65\textwidth,page=4]{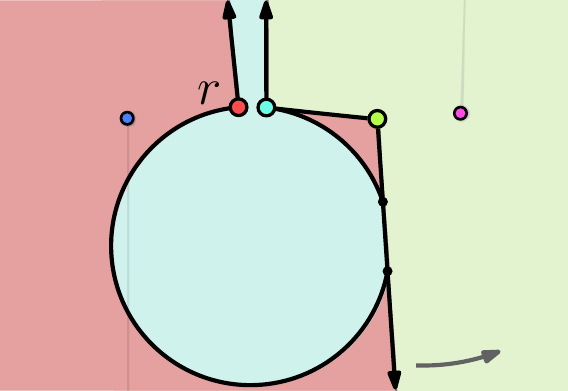}
        \caption{
            A subset of $3$ rays clockwise $(\pi/2)$-rotated. $\rreg(r)$ has 2 faces; to become connected, less rotation is needed.} 
        \label{fig:noangle_example2}
    \end{subfigure}  
    \caption{
        An example of a (thin) polygon with $5$ vertices which justifies \cref{remark:noavd}.
        }
  \label{fig:noangle_example}
\end{figure}

\begin{remark}
    \label{remark:noavd}
    There are sets of rays $\cR_{\Pol}$ for which there exists no unique 
    angle to rotate the rays 
    so that axiom (A2) is satisfied.
\end{remark}

Note that we partitioned $\cR_\Pol$ in a way such that any two rays $r$ and $s$ in a set $\cR^r_d$ have an angular difference of at most $\pi/2$, i.e., $\min \{ \diff(r,s), \diff(s,r) \} \leq \pi/2$.
This is a key property in proving the following lemma.

\begin{lemma}\label{lem:polySubsetStructure}
  The system of bisectors of $\cR^r_d$ satisfies the axioms (A1)-(A3).
\end{lemma}

\begin{proof}
We prove each of the three axioms separately. 

\textbf{(A1):}
Let $r,s$ be a pair of rays in $\cR^r_d$.
We show that
$r$ and $s$ do not intersect. Then, by \cref{lem:NonParallelBisector}, the bisector of two non-intersecting rays is an unbounded simple curve.

Let $x \in r \setminus \{p(r)\}$ (resp. $y \in s \setminus \{p(s)\}$), denote a point lying on $r$ (resp. $s$), and let $L$ be the line passing through $\apex(r)$ and $\apex(s)$; see \cref{fig:axioms_proof1}.
Due to the convexity of the polygon, it follows that $\angle(x,\apex(r),\apex(s))$ and $\angle(\apex(s),\apex(r),x)$ are greater than or equal to $\pi/2$; hence $r$ lies in the closed halfplane orthogonal to $L$ incident to $\apex(r)$ which does not contain $\apex(s)$.
Analogously $s$, lies in the closed halfplane orthogonal to $L$ incident to  $\apex(s)$ that does not contain $\apex(r)$.
Thus, the horizontal strip defined by the two halfplanes (shown shaded in \cref{fig:axioms_proof1}) separates $r$ and $s$, and so, they do not intersect.

\textbf{(A2):}
It suffices to prove the property for any subset of $\cR^r_d$ of size three~\cite{klein2009}.
By \cref{lem:oneUnbounded}, each Voronoi region has exactly one unbounded face, so if a region is disconnected, then it must have  at least one bounded face.
The diagram of three rays can have at most one proper Voronoi vertex, as shown in \cref{lem:diagramOfThree}.
Thus, a bounded face in the diagram cannot be bounded only by edges incident to proper vertices, and so, such a bounded face must appear incident to a ray.
Hence, it suffices to show that no ray intersects twice the bisecting circle of the other two rays.

Let $\{ r,s,t \}$ be a subset of $3$ rays of $\cR^r_d$.
  We will show that no ray in $\{r,s,t\}$ intersects twice the bisecting circle defined by the other two rays in $\{r,s,t\}$.  
  We prove this for $r$, i.e., $r$ does not intersect twice $\cb(s,t)$; the cases of $s,t$ are analogous.
  Without loss of generality we assume that $\diff(t,s)<\diff(s,t)$. We divide the proof into three cases, depending on the relative position of $\apex(r)$. 
  
  (1) If $\apex(r)$ lies in the interior of $\cb(s,t)$, then $r$ intersects $\cb(s,t)$ exactly once and the claim follows.  
  
  (2) Assume that $\apex(r)$ appears between $\apex(s)$ and $\apex(t)$ along the convex polygon chain, and let $I$ be the intersection of the supporting lines $l(s)$ and $l(t)$;
  see \cref{fig:axioms_proof2}. 
  Because of the rotation of rays, we have $\angle(\apex(r),\apex(s),I)\leq\pi/2$ and $\angle(I,\apex(t),\apex(r))<\pi/2$.
  Thus, by the properties of cyclic quadrilaterals\footnote{A cyclic quadrilateral or inscribed quadrilateral is a quadrilateral whose vertices lie on a single circle, and therefore, the four perpendicular bisectors (of the sides) are concurrent. Also, a convex quadrilateral is cyclic if and only if its opposite angles are supplementary (i.e., their sum is $\pi$).},
  $\apex(r)$ lies in the interior of $\cb(s,t)$, the circle through $\apex(s)$, $\apex(t)$ and $I$, 
  see that $\alpha+\beta>\pi$ in \cref{fig:axioms_proof2}). Thus, $r$ intersects $\cb(s,t)$ once. 
  
\begin{figure}[t]
  \centering
  \begin{subfigure}[t]{0.325\textwidth}
  \includegraphics[width=\textwidth,page=9]{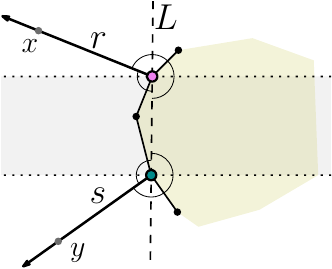}
    \caption
    {
      (A1): The rays $r$ and $s$ are separated by the shaded corridor.
    }
    \label{fig:axioms_proof1}
  \end{subfigure}
  \hfill
  \begin{subfigure}[t]{0.325\textwidth}
  \includegraphics[width=\textwidth,page=11]{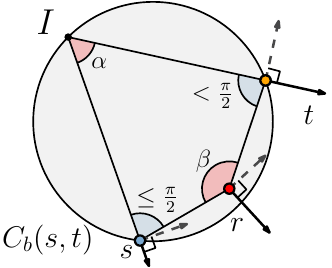}
    \caption
    {
      (A2): $\apex(r)$ is between $\apex(s)$ and $\apex(t)$, so it lies in $\cb(s,t)$.
      }
    \label{fig:axioms_proof2}
  \end{subfigure}
  \hfill
  \begin{subfigure}[t]{0.325\textwidth}
    \includegraphics[width=\textwidth,page=12]{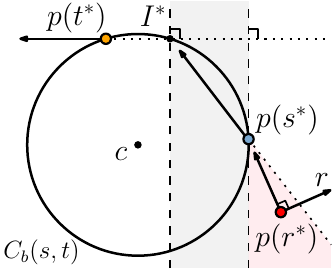}
    \caption
    {
      (A2): The vertical strip (shown shaded) separates $c$ and $r$.} 
    \label{fig:axioms_proof3}
  \end{subfigure}
  \caption
  {
    Illustrations for the proof of \cref{lem:polySubsetStructure}.
  }
  \label{fig:axioms_proof}
\end{figure}
  
  (3) It remains to study the case where $\apex(r)$ lies outside $\cb(s,t)$, and $\apex(r)$ appears before or after both $\apex(s)$ and $\apex(t)$, in counterclockwise order.
  We consider the case when $\apex(r)$ appears before $\apex(s)$ and $\apex(t)$.
  The other case is analogous. 
  
  Let $c$ denote the center of $\cb(s,t)$, let $r^*$ (resp. $s^*$, $t^*$) denote the ray $r$ (resp. $s$, $t$) rotated counterclockwise by $\pi/2$ around its apex, and let $I^*$ denote the intersection point between the supporting lines $l(s^*)$ and $l(t^*)$.
  Without loss of generality, we assume that $t^*$ is a horizontal ray pointing left;
  refer to \cref{fig:axioms_proof3}.
  If $r$ intersects twice $\cb(s,t)$, then $c$ lies to the right of the (directed) line $l(r^*)$,
  so it suffices to prove that $c$ lies to the left of $l(r^*)$.
  To show this, observe that $\apex(r^*)$
  lies to the right of the vertical line through $\apex(s^*)$ (since $\diff(t^*,r^*) \leq \pi/2$), and to the left of $l(s^*)$ (since $r^*$ and $s^*$ are induced by a convex polygon); see the red region in \cref{fig:axioms_proof3}.
  On the other hand, the vertical line through $I^*$ separates $C$ from $\apex(r^*)$.
  Hence, $c$ lies to the left of $r^*$, proving the claim.

\textbf{(A3):}
The diagram $\rvd(\cR^r_d)$ is defined by distance functions, one for each site in $\cR^r_d$, whose domain is the entire plane.
Hence, any point in the plane must belong to the closure of some region of $\rvd(\cR^r_d)$.
\end{proof}

Since each Voronoi region is connected and unbounded, and since $\rvd(\cR)$ is connected (\cref{lem:connectedDiagram}), we can infer the following.

\begin{corollary}
    \label{cor:rvdConvexDiagram4complexity}
  $\rvd(\cR^r_d)$ is a tree of $\Theta(\cR^r_d)$ complexity.
\end{corollary}

The results on abstract Voronoi diagrams \cite{klein1989} directly imply a randomized $O(n\log n)$ time algorithm to construct $\rvd(\cR^r_d)$.
We can further improve 
this time complexity, to $O(n)$-time, by showing that the system of bisectors of $\cR^r_d$ falls under the more restricted \emph{Hamiltonian abstract Voronoi diagram} framework \cite{klein1994}.
In addition to satisfying axioms (A1)-(A3), the following axiom must also be satisfied: 
\begin{itemize}
    \item \textbf{(A4)}
  There exists a simple curve $\Ham$ of constant complexity such that $\Ham$ visits each region $\rreg(r)$ in $\rvd(\cR')$, $\forall \cR' \subseteq \cR$ and $\forall r \in \cR'$, exactly once.
  $\Ham$ can be closed or unbounded.
\end{itemize}

If the system of bisectors of $\cR$ satisfies axioms (A1)-(A4) and the ordering of the regions of $\rvd(\cR')$ along $\Ham$ is known $\forall \cR' \subseteq \cR$, 
then $\rvd(\cR)$ can be computed in $\Theta(n)$-time~\cite{klein1994}.
Hence for our problem, it suffices to ]
find a curve $\Ham$ satisfying these properties.

\begin{lemma}
\label{lem:subsetAlgorithm}
    $\rvd(\cR^r_d)$ can be constructed in deterministic $\Theta(\vert\cR_d\vert)$ time. 
\end{lemma}

\begin{proof}
    We show that $\cR^r_d$ satisfies axiom (A4) as defined above;
    the linear time algorithm is then a direct corollary of the existing results~\cite{klein1994}.
    
    Let $\Ham$ be a circle of sufficient large radius, such that all bisecting circles lie entirely in the interior of $\Ham$.
    For any $\cR' \subseteq \cR^r_d$, the diagram $\rvd(\cR')$ is a tree (\cref{cor:rvdConvexDiagram4complexity}), so its faces are all unbounded. 
    By its definition, $\Ham$ does not intersect any bisecting circle, hence, $\Ham$ must visit
    each region of $\rvd(\cR')$ exactly once; a change in the visited region takes place when $\Ham$ intersects a ray.
    
    The ordering of the unbounded faces of $\rvd(\cR^r_d)$ corresponds to 
    the ordering of the respective vertices along the polygon $\Pol$, and this is maintained for any $\cR' \subset \cR^r_d$.
    Further, the ordering of the vertices of $\Pol$ is part of the input, concluding the proof.
\end{proof}

\subsection{\texorpdfstring{$\Theta(n)$}{O(n)}-time algorithm: merging the four diagrams}
\label{subsec:pvdMerging}

We now merge all four diagrams to obtain $\pvd(\cR_{\Pol})$.
Our merging process consists of two phases.
In the \emph{first phase} we merge $\rvd(\cR^r_{W})$ and $\rvd(\cR^r_{S})$ to obtain $\rvd(\cR^r_{W} \cup \cR^r_{S})$; respectively for 
$\rvd(\cR^r_{E} \cup \cR^r_{N})$; see \cref{fig:convex_instance_merge_1}.
In the \emph{second phase} we merge the diagrams $\rvd(\cR^r_{W} \cup \cR^r_{S})$ and $\rvd(\cR^r_{E} \cup \cR^r_{N})$, restricted to the interior of $\Pol$, to obtain $\pvd(\cR_{\Pol})$; see \cref{fig:convex_instance_merge_2}.
We first outline the merging process 
at high level and then 
delve into the details of each procedure separately.
The process requires attention because the resulting diagrams do not fall under the  framework of abstract Voronoi diagrams. 
We will ultimately prove the following.

\begin{figure}[b]  
    \centering
    \begin{subfigure}[t]{0.325\textwidth}
        \includegraphics[width=\textwidth, page=9]{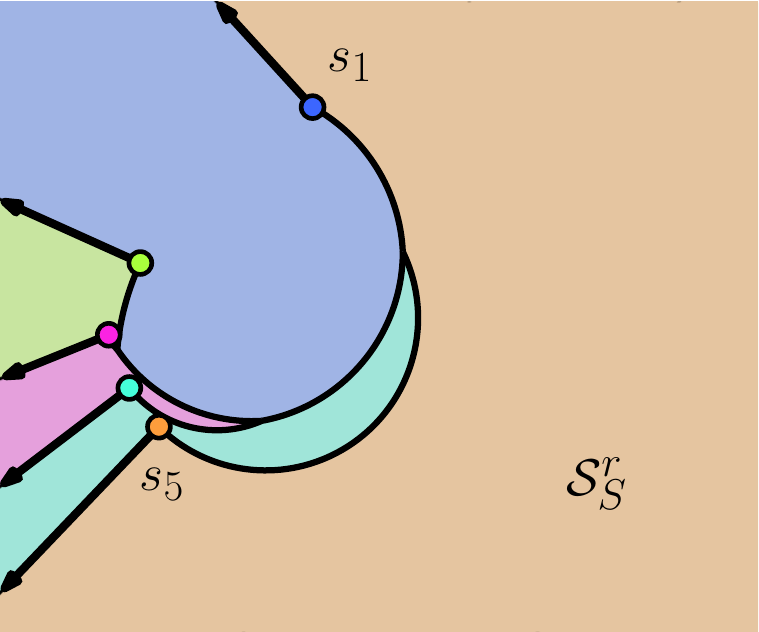}
        \caption{
           Diagram $\rvd(\cR^r_{S})$.
        }
    \label{fig:convex_instance_merge_1one}
    \end{subfigure}
    \hfill
    \begin{subfigure}[t]{0.325\textwidth}
        \includegraphics[width=\textwidth, page=10]{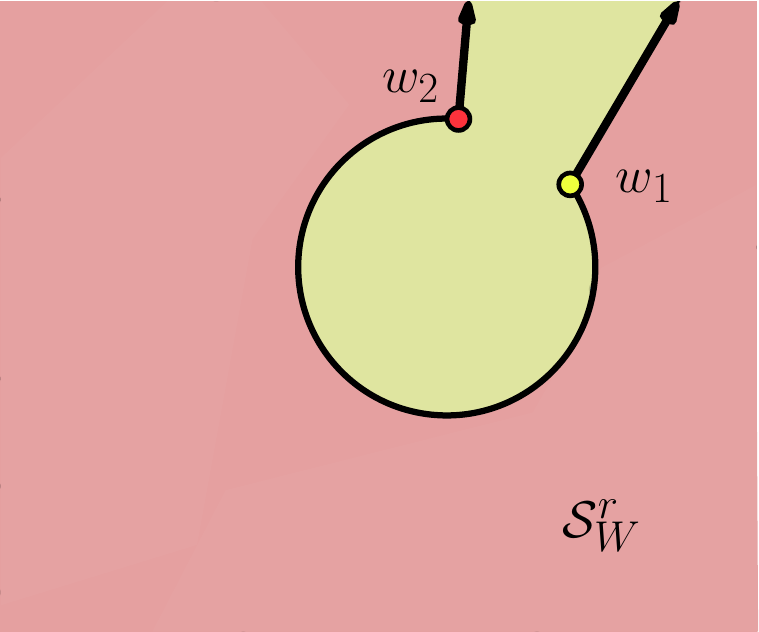}
        \caption{
           Diagram $\rvd(\cR^r_{W})$.
        }
    \label{fig:convex_instance_merge_1two}
    \end{subfigure}
    \hfill
    \begin{subfigure}[t]{0.325\textwidth}
        \includegraphics[width=\textwidth, page=11]{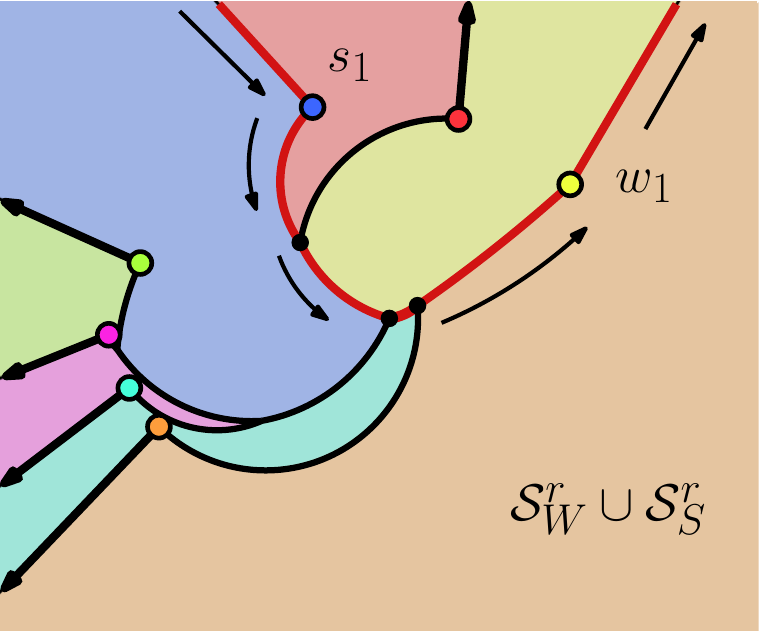}
        \caption{
           Diagram $\rvd(\cR^r_{W} \cup \cR^r_{S})$.
        }           
    \label{fig:convex_instance_merge_1both}
    \end{subfigure}  
    \caption{
        First merging phase: merging $\rvd(\cR^r_{W})$ and $\rvd(\cR^r_{S})$.
        The highlighted red edges correspond to~the merge curve,
        and the black arrows schematize tracing which starts at infinity along ray $s_1$ and ends at ray $w_1$. 
    }
    \label{fig:convex_instance_merge_1}
\end{figure}

\begin{lemma}
    \label{lem:polyMerge}
    Given $\rvd(\cR_d)$ for all $d \in \{\textup{N,W,S,E}\}$, we can merge the four diagrams to obtain $\pvd(\cR_{\Pol})$ in $\Theta(n)$ time.    
\end{lemma}

\subparagraph{Outline of the merging process.}
We describe how to merge $\rvd(\cR^r_{W})$ with $\rvd(\cR^r_{S})$; merging $\rvd(\cR^r_{E})$ with $\rvd(\cR^r_{N})$ is analogous.
Let $w_1,\dots,w_k$ be the rays in $\cR^r_W$ and let $s_1,\dots,s_l$ be the rays in $\cR^r_S$ 
as they appear in counterclockwise order along the boundary of  $\Pol$. The rays in  $\cR^r_E$ and $\cR^r_N$ in the same order are denoted $e_1,\dots,e_p$ and $n_1,\dots,n_q$, respectively.

We need to construct the \emph{merge curve} of the two diagrams, which partitions $\R^2$ in two parts, and keep from one side the diagram $\rvd(\cR^r_{W})$ and from the other 
$\rvd(\cR^r_{S})$; refer to \cref{fig:convex_instance_merge_1both}, where the red edges illustrate the merge curve.
In the first phase, the merge curve consists of the two rays $s_1$ and $w_1$, and the set of circular edges of $\rvd(\cR^r_{W} \cup \cR^r_{S})$ equidistant to sites $w \in \cR^r_W$ and $s \in \cR^r_S$. 
The set of circular edges forms a single connected chain bounded by $\apex(s_1)$ and $w_1$.
We denote the set of circular edges in a merge curve by $E_C$.

In the second phase, we only perform merging inside the polygon $\Pol$, merging  $\rvd(\cR^r_{W} \cup \cR^r_{S})$ and  $\rvd(\cR^r_{E} \cup \cR^r_{N})$ restricted within $\Pol$.
Since the computation is restricted in $\Pol$, the merge curve consists only of the circular edges in $E_C$, which is a single connected chain bounded by $\apex(e_1)$ and $\apex(w_1)$; see \cref{fig:convex_instance_merge_2both}.

Following, we describe in detail the merging process and prove the correctness of our statements. Constructing the merge curve is based on finding a starting point along the merge curve, and then \emph{tracing} it, as in standard Voronoi diagram of points 
(see e.g., \cite{aurenhammerBook}).

\begin{figure}[t]  
    \centering
    \begin{subfigure}[t]{0.325\textwidth}
        \includegraphics[width=\textwidth, page=12]{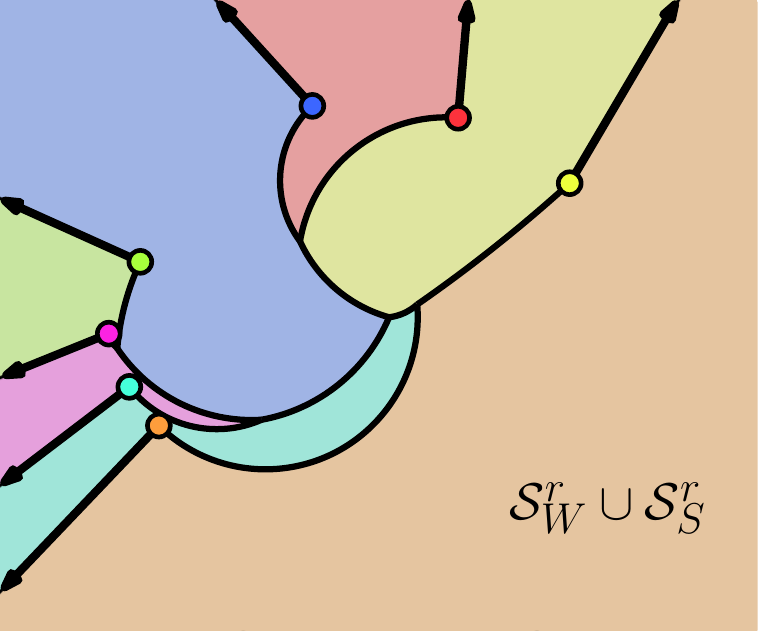}
        \caption{
           Diagram $\rvd(\cR^r_{W} \cup \cR^r_{S})$.
        }
    \label{fig:convex_instance_merge_2one}
    \end{subfigure}
    \hfill
    \begin{subfigure}[t]{0.325\textwidth}
        \includegraphics[width=\textwidth, page=13]{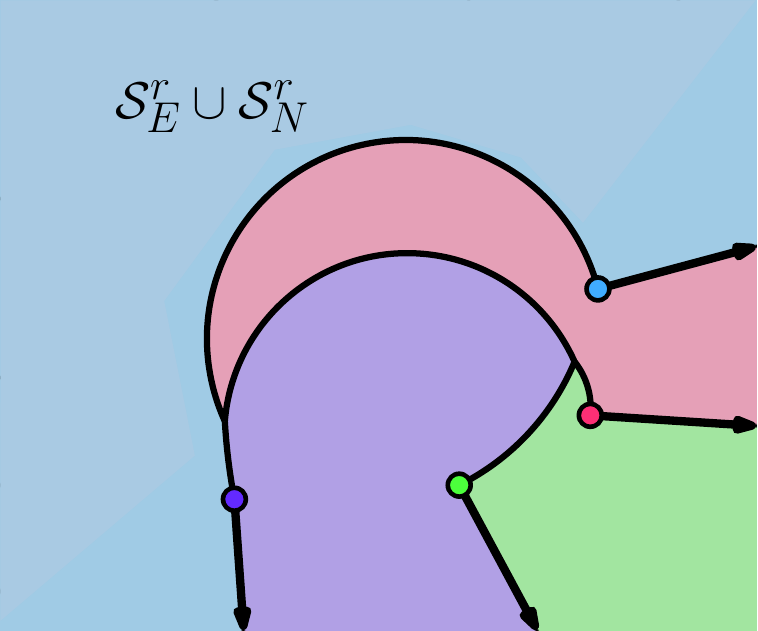}
        \caption{
           Diagram $\rvd(\cR^r_{E} \cup \cR^r_{N})$.
        }
    \label{fig:convex_instance_merge_2two}
    \end{subfigure}
    \hfill
    \begin{subfigure}[t]{0.325\textwidth}
        \includegraphics[width=\textwidth, page=14]{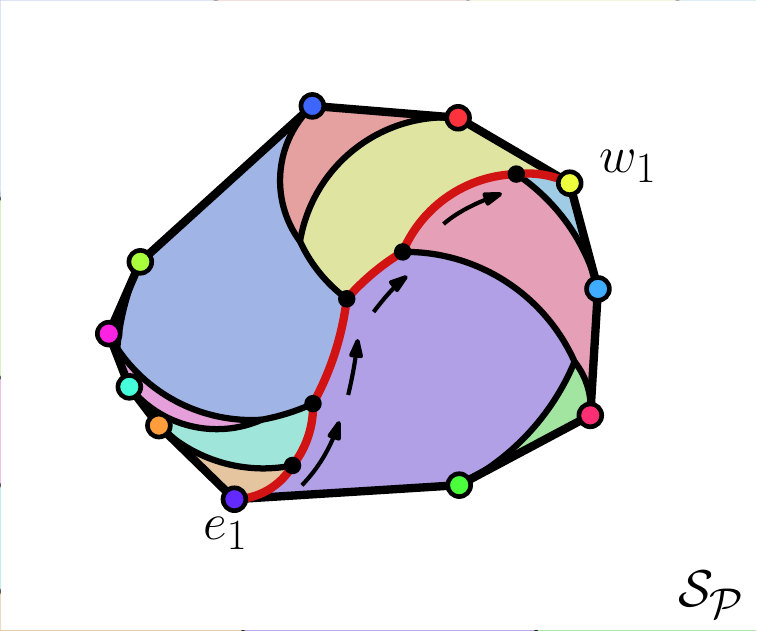}
        \caption{
           Diagram $\pvd(\cR_{\Pol})$.
        }           
    \label{fig:convex_instance_merge_2both}
    \end{subfigure}  
    \caption{
        Second merging phase: merging $\rvd(\cR^r_{W} \cup \cR^r_{S})$ and $\rvd(\cR^r_{E} \cup \cR^r_{N})$ restricted to $\Pol$.
    }
    \label{fig:convex_instance_merge_2}
\end{figure}

\subparagraph{Tracing along the rays (first merging phase).}
    As already mentioned, the merge curve at the first merging phase consists of the two rays $s_1$ and $w_1$ and the set of circular edges $E_C$.
    We start tracing the merge curve, by a point on the ray $s_1$ at infinity.
    Tracing along $s_1$ can be done trivially, this is because the ray lies entirely in $\rreg(w_k)$.
    To see that, consider the set $\cR_W$ (before rotation) and continuously clockwise rotate all rays by an angle of $\pi/2$.
    During this process, $w_k$ does not intersect any of the rays in $\cR_W$, hence $s_1 \in \rreg(w_k)$.

    After ray $s_1$, tracing continues along the circular edges $E_C$ (described in the next paragraph) and finally it reaches ray $w_1$.
    Tracing along the ray $w_1$ is done in a different way.
    In contrast to $s_1$, the ray $w_1$ may intersect many circular edges of $\rvd(\cR^r_S$), each inducing a vertex on $w_1$; see e.g., \cref{fig:convex_merge_special2}.
    To identify such vertices, we intersect $w_1$ with $\rvd(\cR^r_S)$.
    This can be easily done in $O(|\cR_S|)$ time, as $\rvd(\cR^r_S)$ is proved to be a tree (see \cref{cor:rvdConvexDiagram4complexity}).
    Further, the curve $E_C$ might intersect $w_1$ at some point other than $p(w_1)$; 
    see e.g., \cref{fig:convex_merge_special1}.
    In this case the aforementioned search 
    should start from that point.
    
\begin{figure}[t]
    \centering
    \begin{subfigure}[t]{0.48\textwidth}
    \includegraphics[width=.93\textwidth,page=13]{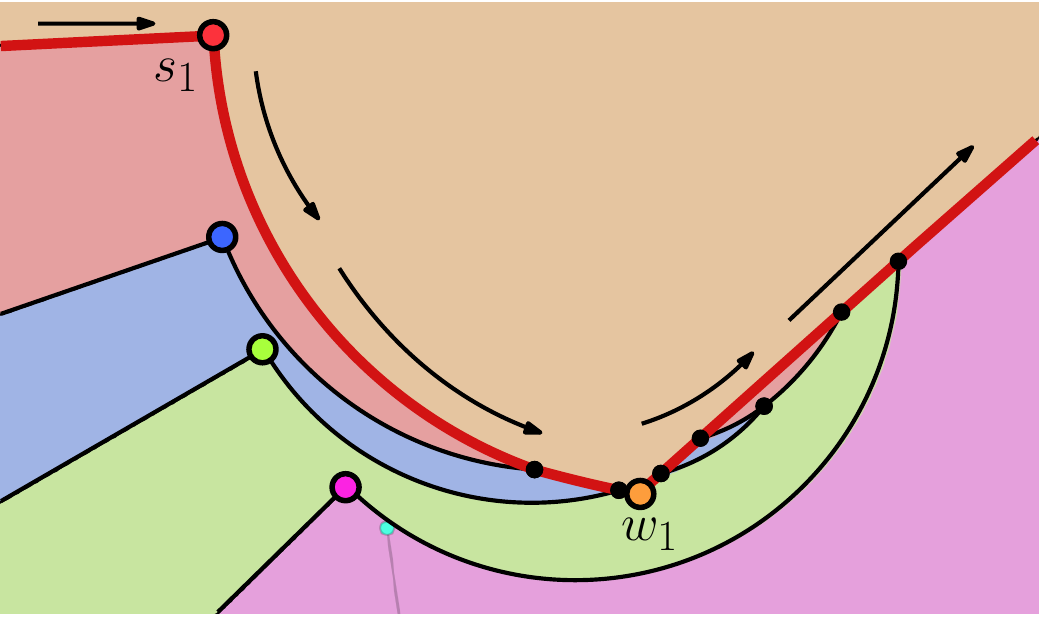}
        \caption{
        Ray $w_1$ intersects $\rvd(\cR^r_{S})$.
        }
        \label{fig:convex_merge_special2}
    \end{subfigure}
    \hfill
    \begin{subfigure}[t]{0.48\textwidth}
        \includegraphics[width=.93\textwidth,page=14]{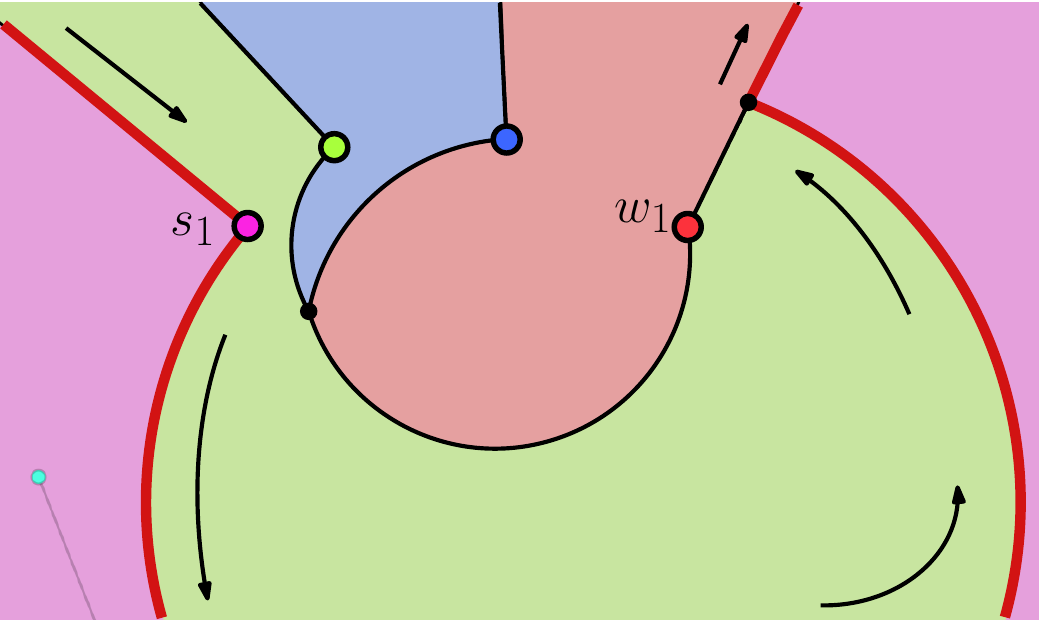}
        \caption{
            The sequence $E_C$ does not end at $p(w_1)$.
        }
        \label{fig:convex_merge_special1}
    \end{subfigure}   
    \caption{
    Two special cases of merging two diagrams $\rvd(\cR^r_{W})$ and $\rvd(\cR^r_{S})$.
    }
    \label{fig:convex_merge_special} 
\end{figure}

\subparagraph{Tracing the sequence of circular edges $E_C$ (both merging phases).}

We now describe how to trace $E_C$ in $\Theta(|E_C|)$ time by adapting 
the standard procedure for tracing a merge curve in Voronoi diagrams, as done for example for bisectors of points \cite{aurenhammerBook}, to angular bisectors.
To establish correctness, 
however, we still need to prove that no \emph{backtracking} needs ever be done during merging, i.e., while tracing any portion of a Voronoi region is scanned at most once;
we prove this in \cref{lem:noBacktrack}.

Suppose we are in the process of merging two ray Voronoi  diagrams, tracing a merge curve, whose main portion of circular edges is denoted by $E_C$. 
Let $L$ be the set of rays defining the diagram to the left of the curve $E_C$ and $R$ the set to the right hand side.
Without loss of generality assume that we are tracing $E_C$ from top to bottom; refer to \cref{fig:rvdConvexTracing1}.
Let $\rreg_L(l)$ denote the Voronoi region of site $l \in L$ within $\rvd(L)$; respectively for $\rreg_R(r)$, $r\in R$.
Suppose that the current edge $b$ of $E_C$ has just entered a region $\rreg_L(l)$ at point $v$.
Let $r$ be the site of the right diagram such that the edge $b$ lies in the region $\rreg_R(r)$.
We determine the points $v_L$ (resp. $v_R$) where $b$ leaves the region $\rreg_L(l)$ (resp. $\rreg_R(r)$).
The point $v_L$ is found by scanning the boundary of $\rreg_L(l)$ clockwise starting from $v$.
The point $v_R$ is found by scanning the boundary of $\rreg_R(r)$ counterclockwise starting from $v$.
Without loss of generality assume that vertex $v_R$ is reached first, which then describes the endpoint of edge $b$.

$E_C$ continues from $v_R$ with another edge $b_2$ along the bisector $\rbis(l,r_2)$, where $r_2$ is another site in $R$.
To determine $v_{L2}$, we scan the boundary of $\rreg_L(l)$, starting from $v_L$ and moving clockwise. This is shown correct in the following lemma, which establishes that $v_{L2}$ cannot be on the boundary of $\rreg_L(l)$ that has already been scanned. Analogously for $v_{R2}$.

\begin{figure}[t]
    \centering
    \begin{subfigure}[t]{0.41\textwidth}
    \centering
    \includegraphics[width=\textwidth,page=1]{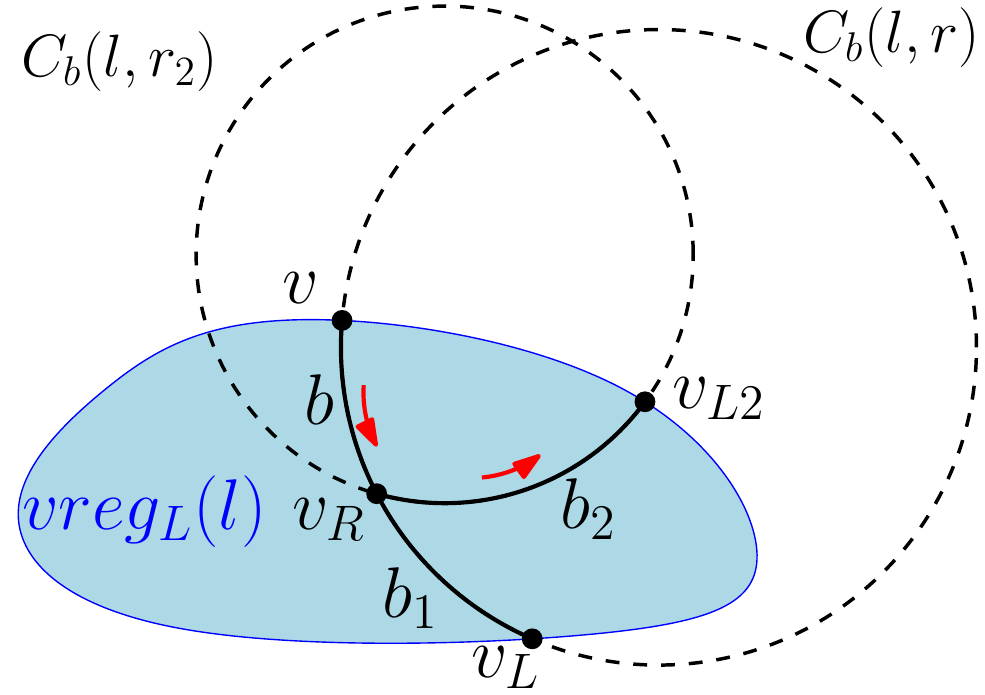}
        \caption{
            Tracing the merge curve (shown with arrows) in a region $\rreg_L(l)$ (shown shaded).
        }
    \label{fig:rvdConvexTracing1}
    \end{subfigure}
    \hfill
    \begin{subfigure}[t]{0.55\textwidth}
    \centering
        \includegraphics[width=.41\textwidth,page=2]{fig/backtracking.pdf}
        \hfill
        \includegraphics[width=.41\textwidth,page=3]{fig/backtracking.pdf}
        \caption{
          Tangents to the circular edges incident to~a proper vertex, and the respective curvilinear~angles.
        }        
    \label{fig:rvdConvexTracing2}
    \end{subfigure}        
    \caption{
    Illustration of the tracing process while merging $E_C$,
    and the non-necessity to backtrack.
    }
    \label{fig:rvdConvexTracing}
\end{figure}

\begin{lemma}
    \label{lem:noBacktrack}
    There is no need to backtrack while tracing the sequence of circular edges $E_C$.
\end{lemma}
\begin{proof}
    Let $v$, $v_L$ and $v_{L2}$ be the points as defined in 
    the above description of tracing $E_C$.  
    To prove that there is no need for backtracking, it suffices to prove that the points $v$, $v_L$ and $v_{L2}$ appear in counterclockwise order along the boundary of face $\rreg_L(l)$;
    see \cref{fig:rvdConvexTracing1}.
    
    The \emph{curvilinear angle} between two intersecting curves is the angle between their two tangents at the point of intersection. 
    Each proper vertex of the diagram, has degree 3, so, the three edges incident to a vertex induce three curvilinear angles; see \cref{fig:rvdConvexTracing2}. Each such curvilinear angle can be seen as the angle at the intersection of two halfplanes. Hence, each such angle is less than $\pi$.
    
    The point $v_R$ is a proper vertex in the merged diagram. Therefore, the curvilinear angle $\angle(v_{L2},v_R,v)$ between the edges $b_2$ and $b$ is less than $\pi$. On the other hand, the angle $\angle(v_L,v_R,v)$ between the edges $b_1$ and $b_2$ is exactly $\pi$ as both edges lie on the same bisector.
    Within the polygon, two related bisectors intersect at most once.
    Thus, the edge $b_2$ has to hit the boundary of $\rreg_L(l)$ after $v_L$ but before $v$ in counterclockwise order.
\end{proof}

Since there is no backtracking required to trace $E_C$, the tracing takes $\Theta(|E_C|)$ time.

\subparagraph{Correctness of the construction of $E_C$.}
    Following, we show the correctness of some claims used earlier without proof.
    More specifically, we show that
    \textbf{$(i)$} the chain $E_C$ is incident to $w_1$, and $\apex(s_1)$ in the first phase, and to $\apex(e_1)$ in the second phase;
    \textbf{$(ii)$} the chain constructed is the complete curve $E_C$, i.e., there are no other connected components left to identify;
    
    \textbf{$(i)$} In the first merging phase, considering tracing the chain $E_C$ starting at $\apex(s_1)$.
    The distance at $p(s_1)$, is exactly $\pi/2$, and it is monotonically increasing.
    Further, consider the polygonal chain $P^*$ consisting of the line segments $\ov{\apex(w_1)\apex(w_2)},$ $\ov{\apex(w_2)\apex(w_3)},\dots,$ $\ov{\apex(w_k)\apex(s_1)},$ $\ov{\apex(s_1)\apex(s_2)}, \dots,$ $\ov{\apex(s_{l-1})\apex(s_l)}$ and the ray $s_l$.
    The distance of any point on $\Pol^*$ to its nearest ray is $\pi/2$.
    Hence, as the distance along the chain $E_C$ is increasing, the only possibility for this chain to end up is at $w_1$.
    Similarly, in the second merging phase, the distance of any point on the complete polygon $\Pol$ to its nearest ray is $\pi/2$, and hence $E_C$ is bounded by $\apex(e_1)$ and $\apex(w_1)$.
  
    \textbf{$(ii)$} To prove this statement we use the disk diagram, defined in \cref{subsec:pvdProperties}.
    By \cref{lem:oneUnbounded}, each region has exactly one unbounded face, so if there exists another connected component in $E_C$, it has to be bounded.
    A second unbounded component would imply that a Voronoi region has two unbounded faces.
    Suppose that the merge curve has another bounded connected component.
    Such a component is bounded entirely by circular edges, and these edges are induced by the respective bisecting circles of the bisectors.
    Since the bisecting circles of the disk diagrams are supersets of the circular edges appearing in $\rvd(\cR^r_{W} \cup \cR^r_{S})$, $\rvd(\cR^r_{E} \cup \cR^r_{N})$ and $\pvd(\cR_{\Pol})$, such a bounded component would also appear in the respective disk diagrams,
    a contradiction to \cref{lem:rvdConvexDiskDiagram}.
    
    Suppose now that the merge curve has a component which is bounded from one side by a ray.
    In the final merging step this is not possible, as on $\pvd(\cR_{\Pol})$ all points on an edge/ray of $\Pol$ belong to the region of the respective ray.
    In the initial merging step 
    (assuming that we merge $\rvd(\cR^r_{W})$ with $\rvd(\cR^r_{S})$),
    for every ray $s_i \in \cR^r_{S}$,  except from $w_1$,
    the complete right side of the ray is incident to $\rreg(s_{i-1})$.
    This can be proved with the same argument used to show that $s_1 \in \rreg(w_k)$ (where $s_1$ is the first ray of $\cR^r_S$, and $w_k$ is the last ray of $\cR^r_W$, respectively).
    As a result, no bounded component of $E_C$ could be incident to a ray.
    Hence, $E_C$ is a single unbounded chain.    

    From the above discussion we can infer that 
    the merge curve does not induce bounded faces in the resulting diagram, 
    except from ray $w_1$ in the first merging phase.
    As a result in the first merged diagram $\rvd(\cR^r_{W} \cup \cR^r_{S})$, each ray $s_i \in \cR^r_S$ can have at most two faces (the unbounded one and the one incident to $w_1$) and $\rreg(w_i)$ is connected for any $w_i \in \cR^r_W$.
    On the contrary in $\pvd(\cR_{\Pol})$, $\rreg(r_i)$ is connected for any $r_i \in \cR_{\Pol}$.

\subparagraph{Overall time complexity.}
In the first step, tracing the rays $s_1$ and $n_1$ 
takes $\Theta(1)$ time, and tracing the rays $w_1$ and $e_1$, takes $\Theta(|\cR_{S}|)$ time and $\Theta(|\cR_{N}|)$ time, respectively.
Tracing the curve $E_C$ takes $O(|\cR_{W}| + |\cR_{S}|)$ time, and $O(|\cR_{E}| + |\cR_{N}|)$ time, respectively; so in total the first step requires $O(n)$ time.
The final step requires $O(n)$ time to trace $E_C$ and $\Theta(n)$ to restrict the diagram into $\Pol$.
So, the overall merging of the four diagrams takes $\Theta(n)$ time.

Putting everything together, we can trivially split $\cR_{\Pol}$ into four sets in $\Theta(n)$ time, we can construct the four diagrams in $\Theta(n)$ time (\cref{lem:subsetAlgorithm}), and we can merge them in $\Theta(n)$ time (\cref{lem:polyMerge}).
So, we can summarize (and re-state) the main result of this section as follows.

\convexAlgo*

\subsection{Brocard illumination of a convex polygon}
\label{subsec:pvdBrocard}

We now turn to the Brocard illumination problem of a convex polygon $\Pol$.
Our goal is to find the Brocard angle of $\Pol$, which is
\begin{align*}
    && \minangle = \max_{x \in \Pol}\min_{r \in \cR_\Pol}\rdis(x,r).
\end{align*}
Observe that the diagram $\pvd(\cR_\Pol)$ is a subset of $\rvd(\cR_\Pol)$, hence \cref{prop:rvdRealizedVertex} applies also in this setting, and so $\minangle$ is realized on $\pvd(\cR_\Pol)$.
However, since the diagram is strictly confined into $\Pol$, the point realizing the Brocard angle, can only lie on a vertex equidistant to 3 rays; see an example in \cref{fig:rvdConvexInstanceBrocard1}.
\begin{figure}[b]
	\centering
	\begin{subfigure}[t]{0.48\textwidth}
		\centering
		\includegraphics[trim=15mm 18mm 20mm 13mm, clip, width=0.75\textwidth, page=18]{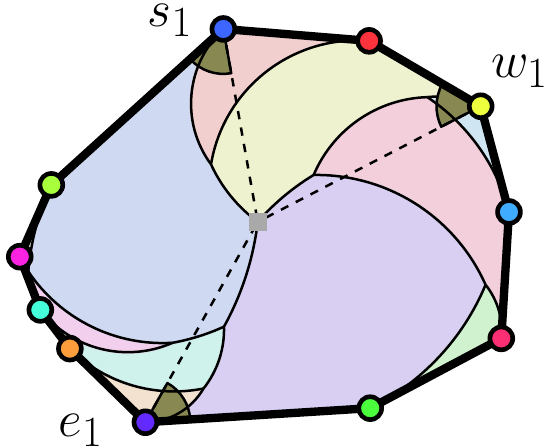} %
		\caption{
			$\pvd(\cR_\Pol)$ and the rays realizing $\minangle$.
		}
		\label{fig:rvdConvexInstanceBrocard1}
	\end{subfigure}
	\hfill
	\begin{subfigure}[t]{0.48\textwidth}
		\centering
		\includegraphics[trim=15mm 18mm 20mm 13mm, clip, width=0.75\textwidth, page=19]{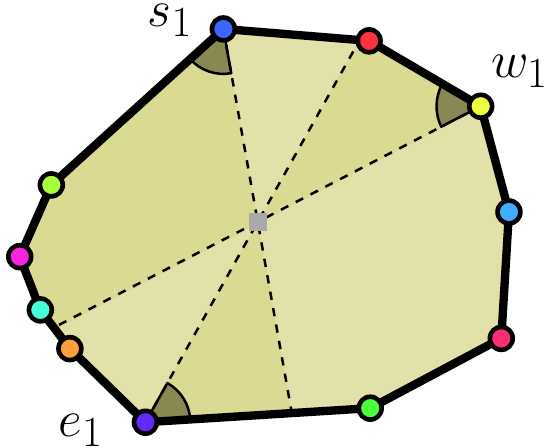}%
		\caption{
			The three $\minangle$-floodlights illuminating $\Pol$.
		}
		\label{fig:rvdConvexInstanceBrocard2}
	\end{subfigure}
	\caption{
		The Brocard angle $\minangle$ of a polygon $\Pol$ realized by $(e_1,w_1,s_1)$.
	}
	\label{fig:rvdConvexInstanceBrocard}
\end{figure}

Similarly to the setting in $\R^2$, to find $\minangle$ we can first construct $\pvd(\cR_{\Pol})$ and then we can traverse it to find the vertex of maximum distance.
Both steps can be done in $\Theta(n)$ time resulting in the following.

\begin{theorem}
	\label{thm:brocard_angle}
	The Brocard angle of a convex polygon $\Pol$  can be found in $\Theta(n)$ time.
\end{theorem}

Following, we give tight bounds on the value of the Brocard angle.

\begin{proposition}
  \label{prop:rvdBoundAngleConvex}
    Given a convex polygon $\Pol$, the range of values of the Brocard angle is $(0,\pi/2 - \pi/n]$.
\end{proposition}
\begin{proof}
  A $\pi/2 - \pi/n$ upper bound on the Brocard angle is given in~\cite{besenyei2015,dmitriev1946}.
  Such an angle is realized by regular polygons.
  The last illuminated point is the center of the polygon, which is simultaneously illuminated by all the floodlights at an angle of $\pi/2 - \pi/n$.
  
  To prove the lower bound note that, while preserving convexity, we can smoothly transform a regular polygon into a polygon whose bounding box has width $w$, height $h$, and an aspect ratio $h/w$ arbitrarily close to zero, so that $\minangle$ is also arbitrarily close to zero.
  Hence it is possible to get any Brocard angle in the range $(0,\pi/2 - \pi/n]$.
\end{proof}

\subparagraph{Illumination of a convex polygon by 3 floodlights.}
Note that the three floodlights which realize $\minangle$ suffice to illuminate $\Pol$, implying the following; see the example of \cref{fig:rvdConvexInstanceBrocard2}.

\begin{remark}
	A convex polygon $\Pol$ can be entirely illuminated by three $\minangle$-floodlights.
\end{remark}

We conclude by discussing an implication of our results.
Consider the following question~\cite{orourke1995}:
given a convex polygon $\Pol$ with $n$ vertices, what is the minimum angle $\beta^*$, such that three vertex $\beta^*$-floodlights (not necessarily aligned with the edges) illuminate $\Pol$.
An $\beta^*=\pi/6$ solution for $n=3$, and an $\beta^*= \pi/4$ solution for $n=4$, is given in \cite{contreras1998}.
Further, a $\beta^*=\pi/3$ solution for arbitrary $n$ is given in \cite{urrutia2000}.
Our results imply a $\beta^*=\minangle$ solution for arbitrary $n$ and, as proved in \cref{prop:rvdBoundAngleConvex}, $\minangle \leq \pi/2 - \pi/n$.
Hence, our results subsume the aforementioned solutions for $n=3,4,6$ and improve the case of $n=5$, to $\beta^*=3\pi/10$.

\section{Rays Voronoi diagram restricted to curves}
\label{sec:curve}

Floodlight illumination problems have also been considered restricted to curves, see e.g., \cite{contreras1998stage,estivill1995,ito1998,toth2003segments}.
Motivated by such problems, 
let $\cR$ be a set of $n$ rays in $\mathbb{R}^2$, and let the domain of interest be a simple curve $\C$.
We denote by $\rvd_\C(\cR)$ the rotating rays Voronoi diagram of $\cR$ restricted to $\C$.
We show that 
$\rvd_\C(\cR)$ can be viewed as the lower envelope of distance functions in 2-space.

\subsection{Brocard illumination of a line}

We first consider the curve $\C$ to be a line; see \cref{fig:line_example}.
We prove the following.

\begin{figure}[b]
    \centering
    \begin{subfigure}[t]{0.49\textwidth}
        \centering
        \includegraphics[width=\textwidth,page=5]{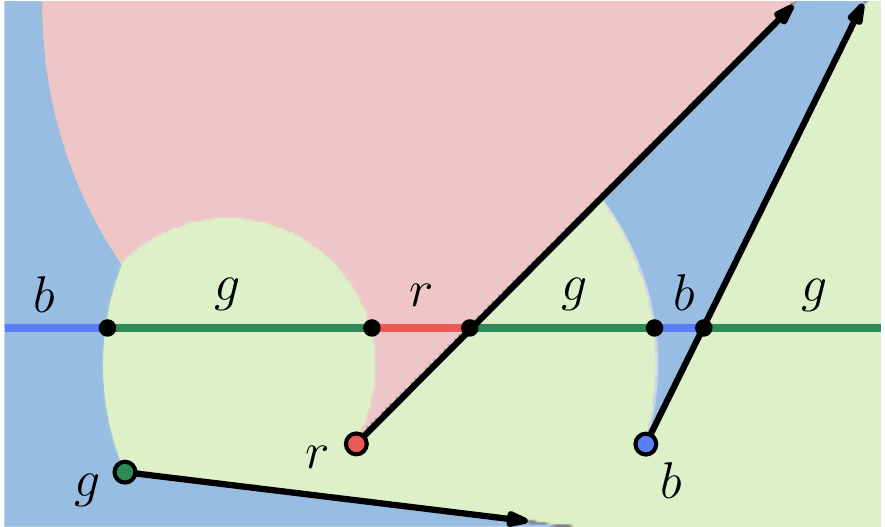}        
        \caption{
            $\rvd_\C(\cR)$ as the intersection of
            $\rvd(\cR)$ in $\mathbb{R}^2$ with $\C$.
            }
    \label{fig:line_example1}
    \end{subfigure}
    \hfill
    \begin{subfigure}[t]{0.49\textwidth}
        \centering
        \includegraphics[width=\textwidth,page=6]{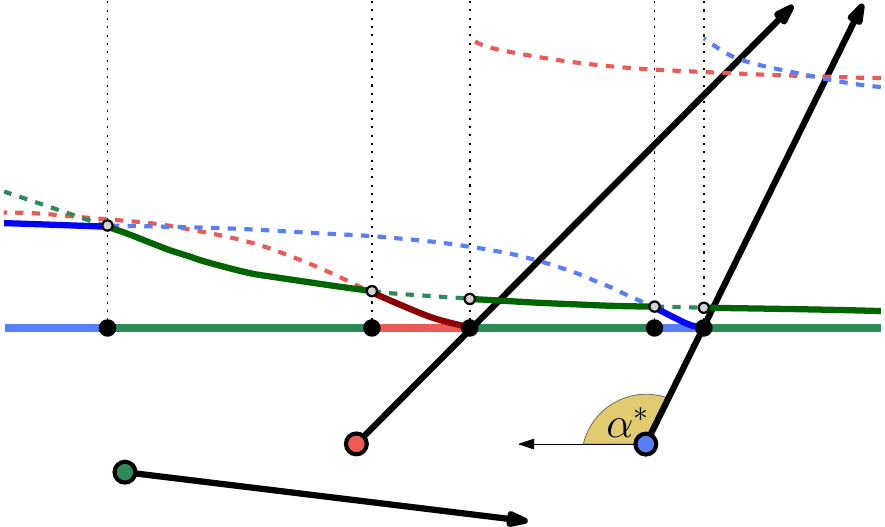}
        \caption{
            $\rvd_\C(\cR)$ as the lower envelope (highlighted) of distance functions (dashed).          
            }
    \label{fig:line_example2}
    \end{subfigure}  
    \caption{
        The curve $\C$ is the horizontal line $x_2=0$, and $\cR$ is a set of 3 rays $\{r,b,g\}$.
    }    
    \label{fig:line_example}
\end{figure}

\begin{theorem}
  \label{thm:rvdCurveLine}
  Given a line $\C$, $\rvd_\C(\cR)$ 
  has complexity $O(n 2^{\alpha(n)})$ and it can be constructed in $O(n\alpha (n)\log n)$ time.
\end{theorem}

\begin{proof}
  Without loss of generality, let $\C$ be the horizontal line $x_2=0$; see \cref{fig:line_example}.
  Each site $r \in \cR$ induces a distance function in 2-space which maps a point $x = (x_1,0) \in \C$ to point $x_{map} = (x_1,\rdis(x,r))$ (dashed curves in \cref{fig:line_example2}).
  Observe that if a ray $r$ intersects $\C$ at point $(i,0)$, there is a point of discontinuity, and the distance function is split into two partially defined functions, one with domain up to $i$ and one with domain starting at $i$.
  The diagram $\rvd_\C(\cR)$ is the lower envelope of all these distance functions projected down to $\C$.
  The lower envelope of $n$ partially defined functions, where each pair of functions intersects at most $s$ times, has $O(\lambda_{s+2}(n))$ complexity~\cite{hart1986} and it can be constructed in $O(\lambda_{s+1}(n)\log n)$ time~\cite{hershberger1989}, where $\lambda_{s}(n)$ is the length of the longest $(n,s)$ \emph{Davenport-Schinzel sequence}.

  Observe that the number of intersections of two distance functions is the same as the number of intersection of their bisecting circle with $\C$.
  In our case, a pair of functions intersects at most twice, as $\C$ may intersect twice the bisecting circle of the two respective rays, 
  so $s=2$.
  Further, we have at most $2n$ partially defined functions.
  Thus, $\rvd_\C(\cR)$ 
  has complexity $O(n 2^{\alpha(n)})$ and it can be constructed in $O(n\alpha (n) \log n)$ time, where $\alpha(n)$ is the inverse Ackermann function.
\end{proof}

Considering the illumination of a line $\C$, given $\cR$, 
the Brocard angle $\minangle$, is realized~at a vertex of $\rvd_\C(\cR)$, or at a point of $\C$ at infinity; see, e.g., in \cref{fig:line_example2}, point $(-\infty,0)$ first illuminated by ray $b$.
So, a simple traversal of $\rvd_\C(\cR)$ reveals $\minangle$ in linear additional time.

\subsection{Brocard illumination of a closed curve}

The aforementioned approach can be generalized to arbitrary simple curves, both bounded and unbounded. We first consider $\C$ to be a closed convex curve, aiming to illuminate the interior of $\C$, i.e., the apices of the rays lie inside $\C$.

\begin{theorem}
    \label{thm:rvdCurveConvexInside} 
    Let $\C$ be a closed convex curve, and let the apices of the rays in $\cR$ lie in the interior of $\C$. 
    Then, $\rvd_\C(\cR)$ has complexity $O(\lambda_{s+2}(n))$ and can be constructed in $O(\lambda_{s+1}(n)\log n)$ time,
    where $s$ is the maximum number of times $\C$ is intersected by a (bisecting) circle.
\end{theorem}
\begin{proof}
  Assume that the curve $\C$ is parametrized in the following form, $\C: [0,1] \to \R^2$ with $\C(0) = \C(1)$.
  Analogously to the approach of \cref{thm:rvdCurveLine}, each site $r$ induces a distance function on the curve $\C$, 
  which maps a value $t \in [0,1]$ to the point $(t,\rdis(\C(t),r))$, and the result immediately follows the results of the envelopes of distance functions in 2-space~\cite{hart1986,hershberger1989}.
\end{proof}

\begin{figure}[b]
    \centering
    \begin{subfigure}[t]{0.325\textwidth}
        \centering
        \includegraphics[width=0.9\textwidth,page=3]{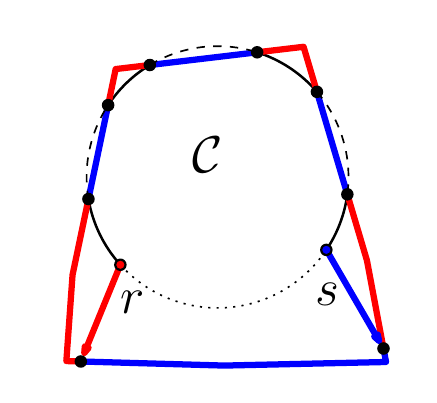}
        \caption{
            Two rays inside $\C$.
        }
        \label{fig:rvdCurveConvexInside}
    \end{subfigure}
    \hfill
    \begin{subfigure}[t]{0.325\textwidth}
        \centering
        \includegraphics[width=0.9\textwidth,page=7]{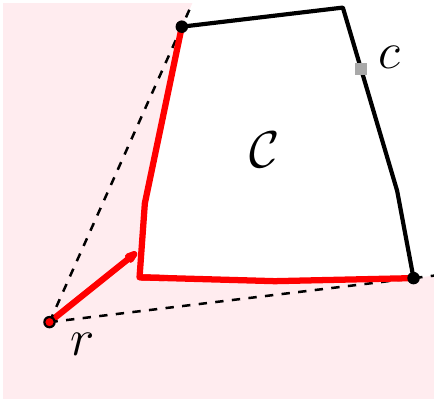}
        \caption{
            The visible portion of a ray. 
        }
        \label{fig:rvdCurveConvexVisible}
    \end{subfigure}
    \begin{subfigure}[t]{0.325\textwidth}
        \centering
        \includegraphics[width=0.9\textwidth,page=6]{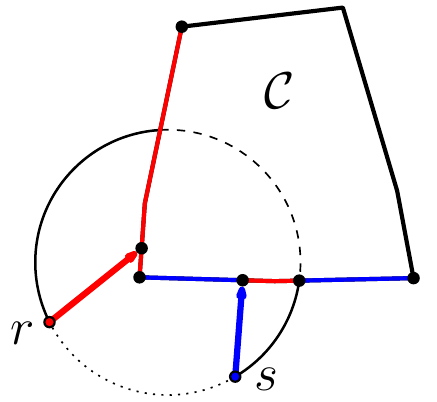}
        \caption{
            Two rays outside $\C$. 
        }
        \label{fig:rvdCurveConvexOutside}
    \end{subfigure}
    \caption{
        $\C$ is a convex polygon.
        The dominance regions of the rays along $\C$ are highlighted.
        }      
\end{figure}

As a corollary, 
if $\C$ is a circle, then $\rvd_\C(\cR)$ has complexity $O(n 2^{\alpha(n)})$ and it can be constructed in $O(n\alpha (n)\log n)$  time, since $\C$ intersects a bisecting circle at most twice; hence, $s=2$.
On the contrary, if $\C$ is an $m$-sided convex polygon, then $\rvd_\C(\cR)$ has complexity $O(\lambda_{2m+2}(n))$ and it can be constructed in $O(\lambda_{2m+1}(n)\log n)$ time, since
$\C$ intersects a bisecting circle at most $2m$ times, hence $s=2m$; see, e.g., the polygon in \cref{fig:rvdCurveConvexInside}.

If we consider the case of illuminating the exterior of a closed convex curve $\C$, i.e., the rays lie outside $\C$, we obtain better results.
In this case, only a part of the curve $\C$ is \emph{visible} by each input ray, where a point $x\in \C$ is visible by a ray $r$ if the open line segment $\ov{\apex(r)x}$ does not intersect $\C$;
see, e.g., in \cref{fig:rvdCurveConvexVisible}, the visible portion of a ray $r$ (point $c$ is not visible by $r$).
If a point $x$ is not visible by a ray $r \in \cR$, we set $\rdis(x,r) = +\infty$.

\begin{figure}[t]
    \centering
    \begin{minipage}[t]{0.4\textwidth}
        \centering
        \includegraphics[width=0.78\textwidth,page=9]{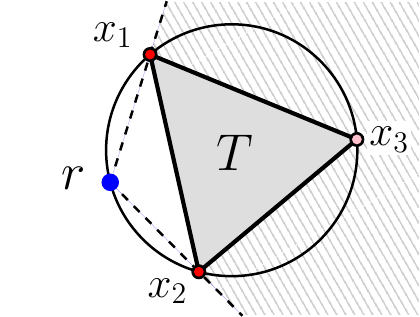}
        \caption{    
        Illustration for the proof of \cref{thm:rvdCurveConvexOutside}.
        Point $x_3$ is not visbile from ray $r$.
        Triangle $T$ is a subset of $\C$.
    }
  \label{fig:rvdCurveProofConvexOutside}
    \end{minipage}
    \hfill
    \begin{minipage}[t]{0.58\textwidth}
        \centering
        \includegraphics[width=0.9\textwidth,page=8]{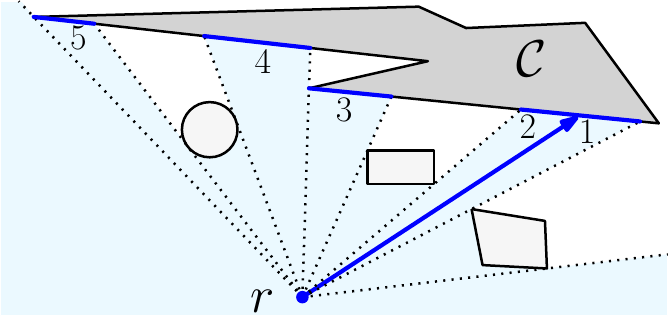}
        \caption{    
        Illumination of $\C$ and the parameter $k$.
        The distance function of a ray $r$ is split into 5 partial functions (k=5) by different types of breakpoints. 
    }
    \label{fig:rvdCurveParameters}
    \end{minipage}
\end{figure}

\begin{theorem}
    \label{thm:rvdCurveConvexOutside} 
    Let $\C$ be a closed convex curve, and let the apices of the rays in $\cR$ lie outside $\C$.  Then, $\rvd_\C(\cR)$ with visibility restrictions has complexity $O(n 2^{\alpha(n)})$ and it can be constructed in $O(n\alpha (n)\log n)$ time.
\end{theorem}
\begin{proof}
    We apply the same approach used as in \cref{thm:rvdCurveLine,thm:rvdCurveConvexInside}.
    To get the claimed results we show the part of the curve visible by any two rays is intersected at most twice ($s=2$) by a bisecting circle. See \cref{fig:rvdCurveConvexOutside}.
    
    Given a ray $r \in \cR$ consider the portion of $\C$ visible by $r$.
    Suppose, for the sake of contradiction, that a bisecting circle defined by $r$ intersects the visible part of $\C$ in at least 3 points $x_1,x_2,x_3$; 
    refer also to \cref{fig:rvdCurveProofConvexOutside}.
    By definition, $\apex(r)$ lies on the bisecting circle, so $\apex(r)$ lies on one of the circular arcs $\ov{x_1x_2}$, $\ov{x_2x_3}$, or $\ov{x_3x_1}$;
    without loss of generality, suppose $\apex(r) \in \ov{x_1x_2}$.
    By the assumption $\apex(r)$ lies outside $\C$, so obviously points $x_1,x_2 \in C$ obstruct the visibility of $r$, and $x_3$ is not visible by $r$.
    But $x_3$ was the intersection point of the bisecting circle with the visible portion of $\C$, a contradiction.

    The part of $\C$ visible by two rays is a subset of the part visible by each of the rays independently, so it is intersected by a bisecting circle at most twice.
    Hence, $s=2$ and as in \cref{thm:rvdCurveLine,thm:rvdCurveConvexInside} the combinatorial and algorithmic results follow.
\end{proof}

\subparagraph{Extensions to other classes of curves.}
Our approach can be extended to illuminate classes of curves which induce visibility restrictions to the rays/sites.
As an example, refer to \cref{fig:rvdCurveParameters}, and consider the illumination of a non-convex polygon $\C$ in the presence of other curves (\textit{obstacles}).
Given a ray $r$, the portion of $\C$ visible by $r$ can be split into many maximal connected components, due to the visibility constraints;
a split might be induced by the curve $\C$ itself (breakpoint (3-4) in \cref{fig:rvdCurveParameters}); it can be induced by other curves (breakpoint (2-3) in \cref{fig:rvdCurveParameters}); or it can be induced by the ray $r$ itself (breakpoint (1-2) in \cref{fig:rvdCurveParameters}).

If the part of the curve $\C$ visible by a ray $r_i$ is split in $k_i$ connected components, this implies that the corresponding distance function of $r_i$ is split into $k_i$ partially defined functions.
Let $K = \sum_{i \in n} k_i$ be the total number of partially defined functions.
Using results on lower envelopes of distance functions in 2-space \cite{hart1986,hershberger1989}, we can derive that $\rvd_\C(\cR)$ has complexity $O(\lambda_{s+2}(K))$ and it can be constructed in $O(\lambda_{s+1}(K)\log K)$ time.

\section{Concluding remarks}
\label{sec:conclusion}

In this work, we studied a new Voronoi structure, the rotating rays Voronoi diagram.
Our motivation for studying this diagram originates from 
the Brocard illumination problem in polygons.
We exhibited a general method for solving the Brocard illumination problem in different domains:
given a domain $D$ and a set of rays $\cR$, we can find the minimum angle $\minangle$ needed to illuminate $D$ using $\minangle$-floodlights aligned with $\cR$, by constructing $\rvd(\cR)$ restricted to $D$.

There are many interesting questions to investigate, both related to the study of $\rvd(\cR)$ as a Voronoi structure, but also related to floodlight illumination problems.
Regarding the $\rvd(\cR)$ in $\mathbb{R}^2$, we would like to settle whether the worst case combinatorial complexity is $\Theta(n^2)$ and whether it can be constructed in $o(n^{2 + \epsilon})$ time.
Regarding the Brocard illumination of polygons, we would like to see how our approach can extend to other classes of (non convex) polygons.
We expect to have difficulties due to the visibility constraints, but we believe that the main concepts of our algorithms can be adapted to work in such settings.

\typeout{}
\bibliography{references.bib}
\end{document}